\documentclass[conference]{IEEEtran}
\IEEEoverridecommandlockouts

\RequirePackage{silence}
\usepackage[utf8]{inputenc} %
\usepackage[T1]{fontenc}    %
\usepackage{microtype}      %

\usepackage{amsmath,amssymb,amsfonts}
\usepackage{amsthm}
\usepackage{hyperref}       %
\hypersetup{
    colorlinks,
    linkcolor={red!50!black},
    citecolor={blue!50!black},
    urlcolor={blue!80!black}
}

\usepackage{cite}
\usepackage{xcolor}
\usepackage{subcaption} %
\usepackage{url}            %
\usepackage{nicefrac}       %
\usepackage{xspace}
\usepackage{xifthen}
\usepackage[probability,operators]{cryptocode}
\usepackage{thmtools,thm-restate}	%
\usepackage{tabularx}
\usepackage{amssymb}
\usepackage{pifont}
\usepackage{multirow}
\usepackage[noend]{algorithmic}
\usepackage{algorithm}
\usepackage{amstext,amssymb,amsmath}
\usepackage{enumitem}
\usepackage{ifsym}
\usepackage{xspace}

\usepackage{amsmath}

\newcommand\overmat[3][0pt]{%
  \makebox[0pt][l]{$\smash{\overbrace{\phantom{%
    \begin{matrix}\phantom{\rule{0pt}{#1}}#3\end{matrix}}}^{\text{#2}}}$}#3}

\newcommand{\parabf}[1]{\smallskip\noindent\textbf{#1}.}
\newcommand{\parait}[1]{\smallskip\noindent\textit{#1}.}

\newcommand{\object}{\ensuremath{o}\xspace}
\newcommand{\secret}{\ensuremath{s}\xspace}
\newcommand{\objectspace}{\ensuremath{\mathbb{O}}\xspace}
\newcommand{\secretspace}{\ensuremath{\mathbb{S}}\xspace}
\newcommand{\guessspace}{\ensuremath{\mathbb{W}}\xspace}

\newcommand{\priors}{\ensuremath{\pi}\xspace}
\newcommand{\prior}[1]{\ensuremath{\pi_{{#1}}}\xspace}
\newcommand{\channel}{\ensuremath{\mathcal{C}}\xspace}
\newcommand{\ChA}{\channel^1}
\newcommand{\ChB}{\channel^2}
\newcommand{\unipriors}{\ensuremath{\upsilon}\xspace}

\newcommand{\system}{\ensuremath{(\pi, \channel)}\xspace}
\newcommand{\nsecrets}{\ensuremath{n}\xspace}

\newcommand{\adversary}{\ensuremath{\mathcal{A}}\xspace}
\newcommand{\advantage}{\ensuremath{\mathsf{Adv}}\xspace}
\newcommand{\distset}[1]{\mathcal{D}(#1)}
\newcommand{\distsetcorner}{\ensuremath{\mathcal{U}\xspace}}

\newcommand{\lnorm}[1]{\ensuremath{L_{{#1}}}\xspace}  %

\newcommand{\TotalVar}{\mathrm{tv}}
\renewcommand{\norm}[1]{\left\lVert #1 \right\rVert}
\newcommand{\Diam}[1]{\mathrm{diam}(#1)}
\newcommand{\ConvHull}[1]{\mathrm{ch}(#1)}
\newcommand{\Reals}{\mathbb{R}}
\newcommand{\Wide}[1]{~#1~}
\newcommand{\eqdef}{\ensuremath{\;\stackrel{\rm def}{=}\;}}

\newcommand{\mulleakage}{\ensuremath{\mathcal{L}^{\times}}\xspace}
\newcommand{\mulcapacity}{\ensuremath{\mathcal{M\!L}^{\times}}\xspace}
\newcommand{\addleakage}{\ensuremath{\mathcal{L}^{+}}\xspace}
\newcommand{\diffp}{\ensuremath{\varepsilon}\xspace}
\newcommand{\bayesrisk}{\ensuremath{R^{*}}\xspace}

\newcommand{\errorguesspriors}{\ensuremath{G}\xspace}

\newcommand{\minbeta}{\ensuremath{\beta^{*}}\xspace}

\newcommand{\minchannel}{\ensuremath{\channel^*}\xspace}
\newcommand{\minprior}{\ensuremath{\priors^*}\xspace}

\newcommand{\cmark}{\ding{51}}%
\newcommand{\xmark}{\ding{55}}%

\DeclareUnicodeCharacter{1F150}{ }

\newcommand{\census}{\texttt{Census1990}\xspace}

\newtheorem{corollary}{Corollary}
\newtheorem{lemma}{Lemma}

\newtheorem{theorem}{Theorem}
\newtheorem{proposition}{Proposition}

\newtheorem{fact}{Fact}

\theoremstyle{definition}

\newtheorem{definition}{Definition}
\newtheorem{game}{Game}

\IEEEoverridecommandlockouts

\begin{document}

\title{Bayes Security: A Not So Average Metric} %

\author{
	\IEEEauthorblockN{Konstantinos Chatzikokolakis\textsuperscript{$\dagger$}}
	\IEEEauthorblockA{University of Athens\\
		kostasc@di.uoa.gr}
	\and
	\IEEEauthorblockN{Giovanni Cherubin\textsuperscript{$\dagger$}}
	\IEEEauthorblockA{Microsoft Research\\
		gcherubin@microsoft.com}
	\and
	\IEEEauthorblockN{Catuscia Palamidessi\textsuperscript{$\dagger$}}
	\IEEEauthorblockA{INRIA, École Polytechnique\\
		catuscia@lix.polytechnique.fr}
	\and
	\IEEEauthorblockN{Carmela Troncoso\textsuperscript{$\dagger$}}
	\IEEEauthorblockA{EPFL\\
		carmela.troncoso@epfl.ch}
}

\IEEEpubid{\begin{minipage}{\textwidth}\ \\[12pt]
		$\dagger$ Equal contribution.
\end{minipage}}

\maketitle

\begin{abstract}
Security system designers favor worst-case security metrics,
such as those derived from differential privacy (DP), due to the
strong guarantees they provide. On the downside, these guarantees
result in a high penalty on the system's performance.
In this paper, we study Bayes security, a security metric
inspired by the cryptographic advantage.
Similarly to DP, Bayes security
i) is
independent of an adversary's prior knowledge,
ii) it captures the worst-case scenario for the two most
vulnerable secrets (e.g., data records);
and iii) it is easy to compose, facilitating security analyses.
Additionally, Bayes security
iv) can be consistently estimated in
a black-box manner, contrary to DP, which is
useful when a formal analysis is not feasible;
and v) provides a better utility-security
trade-off in high-security regimes because it quantifies the risk
for a specific threat model as opposed to threat-agnostic
metrics such as DP.\\
We formulate a theory around Bayes security, and we provide
a thorough comparison with respect to well-known metrics,
identifying the scenarios where Bayes Security is advantageous for
designers.
\end{abstract}


\maketitle

\section{Introduction}\label{sec:intro}

Quantifying the level of protection given by security
and privacy-preserving mechanisms is a fundamental
process in secure system engineering.
To perform a quantitative analysis, one needs to define
appropriate metrics that capture the adversary's gain,
and, ultimately,
what are the risks for the system's users.

A common way of evaluating threats in security and privacy applications
is to quantify
the probability that an adversary guesses some secret
information; this metric is referred to as the
\textit{success rate} or \textit{accuracy} of an attacker.
For example, membership inference attacks against machine
learning (ML) models~\cite{shokri2017membership},
where the attacker aims at guessing
if a data record was used for training the model,
have been for long evaluated w.r.t. the attacker's accuracy.

\parabf{Average-case metrics}
The success rate (or accuracy) has a very clear interpretation:
it measures the probability that an adversary succeeds in the attack.
An important special case, the \textit{Bayes vulnerability} (e.g., \cite{smith2009foundations}), is the
accuracy of the (Bayes) optimal adversary, who has maximal information about the underlying uncertainty.
Both Bayes vulnerability and accuracy rely on the
the \textit{prior} probability of the secret information
that the attacker is trying to guess;
unfortunately, this can result
in misleading conclusions about an attack's strength~\cite{smith2009foundations}.
In the membership inference example, if the prior probability
that a data record is low (say, $0.1$),
a strawman attack that always guesses ``non-member''
will achieve $90\%$ accuracy;
yet, this is a rather weak attack.
This shows that the accuracy metric does not characterize well the risk of this attack.

An alternative criterion, used to evaluate cryptographic primitives,
is the \textit{advantage} (e.g.,~\cite{bellare2005introduction}).
Advantage defines the prior probability over the secrets
to be uniform, by construction, and it relates
this prior probability to the
probability that the adversary succeeds after having
access to the model (accuracy).
\textit{Intuitively}, this metric disregards the contribution of
the prior, and it quantifies the information
leakage of the algorithm itself;
however, to the best of our knowledge, no known result
shows that this
metric is
prior-independent.

Both cryptographic advantage and Bayes vulnerability are \textit{threat-specific},
i.e., they are connected to the threat model under which security is quantified.
This gives a precise interpretation of what attacks they protect against.
However, they are rarely used to study complex real-world
systems, such as ML training algorithms.
The main reason is that, due to complexity, one often
needs to evaluate the security of individual parts of the algorithm and
then \textit{compose them};
this is not known to be possible with these metrics.

\parabf{Worst-case metrics}
At the other end of the spectrum,
Differential Privacy (DP)
has become the golden standard in privacy analysis~
\cite{dwork2006differential}.
In DP, a parameter $\varepsilon$ bounds the probability
that an algorithm's output \textit{leaks any information}.
There are several reasons why DP is generally preferred
over other metrics:
1) DP is easy to compose analytically;
e.g., if two
algorithms are resp. $\varepsilon_1$- and $\varepsilon_2$-DP,
their cascade composition is $(\varepsilon_1+\varepsilon_2)$-DP.
2) DP is \textit{prior-independent}:
it measures the risk of releasing a secret
via the algorithm, independently of the secret's prior
probability;
3) DP protects against virtually any threat model: its guarantees hold whether the
adversary wishes to learn an entire data record or just
one bit of information;
we refer to this property as being \textit{threat-agnostic}.
4) DP considers the \textit{worst-case} scenario over the outputs, ensuring
robustness against any threat: it bounds the best gain an adversary can have, even if their maximum gain
is achieved with negligible probability.

DP, however, also comes with disadvantages.
First, DP is often too strict of a requirement:
in many security settings, such as traffic analysis, side
channel protection, and privacy-preserving ML (PPML), DP mechanisms that provide high
protection levels incur severe utility loss.
This mostly comes from the fact that DP is threat-agnostic.
Second, it is theoretically impossible to estimate empirically
$\varepsilon$-DP in a consistent manner (e.g.,~\cite{cherubin2019fbleau}) for black-box mechanisms;
for example, an empirical DP estimate would fail to properly
assess a mechanism that violates $\varepsilon$-DP
with negligible probability.

\parabf{Bayes security}
In this paper, we study the multiplicative risk leakage ($\beta$),
a metric that generalizes
the cryptographic advantage defined for a
Bayes-optimal adversary~\cite{cherubin2017bayes}.
Specifically,
we focus on the minimizer of $\beta$ across all prior
probability distributions, which we call
the \textit{Bayes security metric} ($\beta^*$).
We identify several properties about $\beta^*$ that
make it suitable
for studying security and privacy threats of complex algorithms:
1)~similarly to the Bayes vulnerability and the advantage,
Bayes security is threat-specific;
2)~differently from them, it quantifies the
risk for the two most vulnerable secrets (\textit{worst-case}).
3)~Similarly to DP,
the \textit{composition} of Bayes secure algorithms is easy to study;
4)~the formal analysis of complex algorithms via Bayes security
is further aided by its direct relation with the total variation
distance between the output distributions of the algorithm;
5)~when a formal analysis is not possible (e.g., one needs to study
a black-box system),
there are consistent methods for estimating Bayes security
(\autoref{sec:fast-minimization}).

Finally, because of its construction, Bayes security captures
the \textit{average} (i.e., expected) risk for the \textit{worst-case}
pair of inputs (e.g., data records, users). In this sense,
it can be regarded as a middle way between average and worst-case
security metrics;
we argue that it gains benefits from both.

We summarize our contributions as follows:

\begin{itemize}[leftmargin=*]

	\item[\cmark] We study multiplicative risk leakage~\cite{cherubin2017bayes},
		$\beta$, a generalization of the advantage.
		We show	that it reaches its least secure setting, $\beta^*$,
		when assigning a uniform prior to the two most vulnerable
		secrets.
		We call this minimum Bayes security, which captures
		the expected risk for the two secrets that are the easiest
		to distinguish for an optimal adversary (\autoref{sec:main-result}).

	\item[\cmark] We study compositionality rules for
		algorithms that satisfy Bayes security (\autoref{sec:composing}).

	\item[\cmark] We study the relation between Bayes security and
		mainstream security and privacy notions.
		We provide a game-based interpretation of Bayes security, which is equivalent to IND-CPA,
		and one for local DP.
		We derive bounds w.r.t. DP and local DP, which
		improve on the bound
		by Yeom et al.~\cite{Yeom:18:CSF}.
		Our analysis shows Bayes security is in between
		worst-case metrics (e.g., DP) and average ones
		(\autoref{sec:other-notions}).
		This enables system designers to explore new
		security-utility trade-offs (\autoref{sec:discussion}).

	\item[\cmark] We derive the Bayes security of three mainstream
		privacy mechanisms: Randomized Response, and the Laplace
		and Gaussian mechanisms (\autoref{sec:case-study}),
		and discuss suitable applications (\autoref{sec:discussion}).

	\item[\cmark] Finally, we provide efficient means to compute $\minbeta$ in white- and black-box scenarios
		(\autoref{sec:fast-minimization}).

\end{itemize}

\noindent
Due to space constraints, proofs are given in the appendix.

\section{Preliminaries}\label{sec:preliminaries}

We consider a \textit{system} \system, where
a \textit{channel} or \textit{mechanism} $\channel$
protects secrets $s \in \secretspace$.
Let $\distset{S}$ be the set of probability distributions over a set $S$.
Secrets are selected as inputs to the channel according to a \textit{prior}
probability distribution $\priors\in\distset{\secretspace}$; we write $\pi_s \eqdef P(s)$.
The channel is a matrix defining the posterior probability of
observing an output $o \in \objectspace$ given an input $s\in\secretspace$:
$\channel_{s,o} \eqdef P(o \mid s)$.
We denote by $\channel_s \in \distset{\objectspace}$ the $s$-th row
of $\channel$ (which is a distribution over $\objectspace$), and
by $\channel_\secretspace$ the set of all rows of $\channel$.
\autoref{tab:symbols-table} summarizes our notation.

\parabf{Adversarial Goal}
We consider a passive adversary $\adversary$ who, given an output $o$, aims at inferring which secret $s$ was input to the mechanism. We model this adversary with the following indistinguishability game,
which we call IND-BAY:

\begin{game}[Indistinguishibility Game]
\centering
\procedure[linenumbering,syntaxhighlight=off]{IND-BAY$_\channel^\adversary$}{%
	\adversary \gets \channel,\priors\\
	s \overset{\priors}{\gets} \secretspace\\
	o \overset{P(o|s)}{\gets} \objectspace\\
	s' \gets \adversary(o)\\
	\pcreturn s = s'
}
\label{game:IND-BAY}
\end{game}

We consider an optimal adversary that has perfect knowledge of
the channel $\channel$ and of the prior distribution over the secret inputs $\pi$ (line 1).
A challenger samples a secret $s$ according
to the prior $\priors$ (line 2), and inputs it to the channel $\channel$
to obtain an observable output $o$ (line 3). For simplicity, the game
considers an individual observation, but we note that sequences of
observations (e.g., representing multiple uses of the same channel to
hide one secret, or simultaneous use of two channels with the same secret)
can be accounted for by redefining $o$ to be a vector.
Upon observing the output $o$.
the adversary produces a prediction $s'$ (line 4).
The adversary wins if they guess the secret correctly: $s=s'$ (line 5).
We evaluate the adversary $\adversary$ with respect to their expected prediction
error according to the 0-1 loss function: $R^\adversary \eqdef P(s \neq \adversary(o)) = P(s \neq s')$.
Extensions to further loss functions are possible, but out of the scope of this paper.

This formulation is different from typical cryptographic games because of the following reasons.
First, we assume an optimal adversary:
instead of providing them with knowledge of the cryptographic algorithm except
for the key, and let them query the primitive to learn its statistical behavior,
we assume that the adversary has perfect knowledge of the probabilistic behavior of the channel.
Second, we compute the advantage with respect to the adversary's error, while cryptographic games
compute the adversary's probability of success.
Third, this game captures an \textit{eavesdropping} adversary that
\emph{cannot} influence the secret used by the challenger to produce the observable output.
This is considered to be a weak adversary in cryptography,
where typically the adversary is allowed to provide inputs to the algorithm under attack.
However, it corresponds to many security and privacy problems where the adversary
cannot influence the secret and only observes channel outputs: website fingerprinting~\cite{hayes2016kfingerprinting,wang2013website},
privacy-preserving distribution estimation~\cite{murakami2018distribution,Murakami019,erlingsson2014rappor}, side channel attacks~\cite{Kocher96,koepf2012automatic,StandaertA08},
or pseudorandom number generation.

\parabf{Adversarial Models} In this paper, we consider the \textit{Bayes adversary}, an idealized adversary
who knows both prior $\priors$ and channel matrix $\channel$, and guesses according to the Bayes rule:
$$s'=\adversary^*(o, \priors, \channel) = \argmax_{\secret \in \secretspace} P(s|o) = \argmax_{\secret \in \secretspace} \channel_{\secret, \object} \prior{\secret} \,.$$
The expected error of the Bayes adversary (\textit{Bayes risk}) is:
$$R^{\adversary^*}(\priors, \channel) = \bayesrisk(\priors, \channel) = 1 - \sum_{\object \in \objectspace}
\max_{\secret \in \secretspace} \channel_{\secret, \object} \prior{s} \,.$$

When this adversary is confronted with a perfect channel,
whose outputs leak nothing about the inputs, their best strategy is to guess according to priors:
$s' = \arg\max_{\secret \in \secretspace} \prior{s}$.
The expected error of this strategy is the \textit{random guessing error}:
$\errorguesspriors(\priors) = 1 - \max_{\secret \in \secretspace} \prior{s}.$
Naturally, $\errorguesspriors(\priors) \geq \bayesrisk(\priors, \channel)$.

\parabf{Multiplicative Bayes risk leakage}
We study the properties of a metric, $\beta$, defined as~\cite{cherubin2017bayes}:
$$\beta(\priors, \channel) \eqdef \frac{\bayesrisk(\priors, \channel)}{\errorguesspriors(\priors)} \,,$$
where it is assumed that $\errorguesspriors(\priors) > 0$;
we refer to $\beta$ as the \textit{multiplicative Bayes risk leakage}.
Inspired from the cryptographic advantage (\autoref{sec:Adv}),
$\beta$ captures how much better than random guessing an adversary
can do. It takes values in $[0,1]$; $\beta=1$ when the system is perfectly
secure (i.e., it exhibits no leakage), and $\beta=0$ when the adversary always
guesses the secret correctly.
In the next section, we define the Bayes security metric $\beta^*(\channel)$
to be the minimizer (i.e., least secure configuration)
of $\beta(\pi, \channel)$ for any prior $\pi$.
We then study its properties, which we argue make it suitable
for analyzing the security of complex real-world algorithms.

We note that $\beta$ is closely related to the
\emph{multiplicative Bayes vulnerability leakage}~\cite{braun2009quantitative}, which
is defined as: $\mulleakage(\pi, \channel) \eqdef \frac{1-\bayesrisk(\priors, \channel)}{1-G(\priors)}$.
Differently from $\beta$, $\mulleakage$ is defined
for the
adversary's probability of \emph{success} (i.e.,
Bayes \textit{vulnerability}) rather than failure,
and it takes values in $[0, n]$.
Despite their similar definition,
these two metrics behave very differently (\autoref{sec:other-notions}).
An important result for
$\mulleakage$
is that it takes its worse value (i.e., least secure configuration)
on a uniform prior over the secrets~\cite{braun2009quantitative}.
Therefore, even when the real priors are unknown, a security analyst can easily
compute a bound on the security of the system.
In the next section, we derive the counterpart result for the
$\beta$: it reaches its least secure configuration when
setting a uniform prior on the
two most vulnerable secrets, and a $0$ prior probability elsewhere;
the proof is substantially more involved
than the one for $\mulleakage$.
We further discuss the relation between the Bayes Security metric and
multiplicative leakage in \autoref{sec:MultLeak}.

\begin{table}[t]
	\centering
	\caption{Notation}
	\begin{tabularx}{\linewidth}{lX}
		\textbf{Symbol} & \textbf{Description} \\
		$\secretspace = \{1, ..., n\}$ & The secret space.\\
		$\objectspace = \{1, ..., m\}$ & The output space.\\
		$\channel_{s, o}$ (abbr. $\channel$) & A channel matrix,
			where $\channel_{s, o} = P(o\mid s)$ for $s \in \secretspace$,
			$o \in \objectspace$.\\
		$\channel_{s}$ & The $s$-th row of a channel matrix. It corresponds
			to the probability distribution $P(o \mid s), \forall o \in \objectspace$.\\
		$\priors \in [0, 1]^n$ & A vector of prior probabilities over the secret space.
			The $i$-th entry of the vector is $\prior{i}$.\\
		$\prior{ij}$ & A prior vector with exactly 2 non-zero entries, in position $i$ and $j$,
			with $i \neq j$.\\
		$\unipriors = (\nicefrac{1}{n}, ... \nicefrac{1}{n})$ & Uniform priors
			for a secret space of size $|\secretspace| = n$.\\
		$\bayesrisk(\priors, \channel)$ (abbr. $\bayesrisk$) & The Bayes risk of a channel.\\
		$\errorguesspriors(\priors)$ (abbr. $\errorguesspriors$) & The random guessing error
			(error when only priors' knowledge is available).\\
		$\beta(\priors, \channel)$ & Bayes security of a channel.\\
		$\minbeta(\channel)$ (abbr. $\minbeta$) & Min Bayes security of a channel.
	\end{tabularx}
	\label{tab:symbols-table}
\end{table}

\section{The Bayes security metric}\label{sec:main-result}

Leakage notions based on the Bayes risk generally depend
on the prior distribution over the secrets.
This makes them unsuitable for measuring
security in real-world applications where the true priors are unknown, e.g., traffic analysis\cite{dyer2012peek-a-boo,wright2007language} or membership inference attacks~\cite{shokri2017membership}, and result in an overestimation of a mechanism's security
if the real prior implies more leakage than the prior considered in the analysis.

Given the similarity between multiplicative \textit{risk} leakage and multiplicative \textit{vulnerability} leakage, one could expect that the uniform
prior also represents the worst case for the latter~\cite{braun2009quantitative}.
Unfortunately, this is not the case: \autoref{thm:bayes-inconsistency} (Appendix~\ref{sec:bayes-inconsistency})
shows that, for secret spaces $|\secretspace| > 2$,
there exists a prior $\priors$ for which $\beta$
is smaller than the one achieved for a uniform prior.

\parabf{Prior-independence for $\beta$}
In this section, we show that the multiplicative Bayes risk
leakage, $\beta(\pi, \channel)$, for a channel $\channel$,
is minimized when the prior $\pi$ assigns equal weight to
the two secrets that are maximally distant (according to
posterior distribution), and 0 to all other secrets.
We refer to this minimizer, representing the highest risk
for the channel w.r.t. adversary's prior knowledge,
as the \textit{Bayes security metric};
we denote it with $\beta^*(\channel)$
(omitting the argument if no confusion arises).
This result makes the Bayes security metric prior-independent:
for any prior knowledge the attacker may have in practice,
$\beta^*$ bounds their success.

For simplicity, we present our result in the \textit{one-try} attack scenario,
as formalized by the IND-BAY game: the adversary observes just \textit{one} output of the system before guessing
the secret input. In Section~\ref{sec:composing} we extend this result
to cases where the adversary can collect more observations.

\begin{restatable}{theorem}{ThmMinBeta}
	\label{thm:minbeta}
	Consider a channel $\channel$ on a secret space with
	$|\secretspace| \geq 2$.
	There exists a prior vector $\minprior \in \distset{\secretspace}$
	of the form
	$$\minprior = \{0, ..., 0, \nicefrac{1}{2}, 0, ..., 0, \nicefrac{1}{2}, 0, ..., 0\}$$
	such that
	$$\minbeta(\channel) = \beta(\minprior, \channel) = \min_{\priors\in\distset{\secretspace}} \beta(\priors, \channel) \,.$$
\end{restatable}

In the following, we provide an intuition of the concepts involved with this
proof. %

We denote with $\distsetcorner^{(k)} \subset \distset{\secretspace}$, for $k = 1, ..., |\secretspace|$,
the set of distributions whose support has cardinality $k$,
and with a uniform distribution over its
non-zero components: %
$$
\distsetcorner^{(k)} \eqdef
\Big\{u \in \distset{\secretspace} \mid u_s \in \{0,\frac{1}{k}\}
\text{ for all } s\in\secretspace \Big\} \,.
$$

For example, if $n = 3$, then:
$\distsetcorner^{(1)} = \{(1, 0, 0), ..., (0, 0, 1)\}$,
$\distsetcorner^{(2)} = \{(\nicefrac{1}{2}, \nicefrac{1}{2}, 0),(\nicefrac{1}{2}, 0, \nicefrac{1}{2}), (0, \nicefrac{1}{2}, \nicefrac{1}{2})\}$,
and $\distsetcorner^{(3)} = \{(\nicefrac{1}{3}, \nicefrac{1}{3}, \nicefrac{1}{3})\}$.

For a fixed channel $\channel$, the proof of \autoref{thm:minbeta} is based on
demonstrating the following two steps:
\begin{enumerate}
	\item the function
		$\beta(\priors, \channel) = \nicefrac{\bayesrisk(\priors, \channel)}{\errorguesspriors(\priors)}$
		has its minimum in the set
		$\distsetcorner = \distsetcorner^{(1)} \cup ... \cup \distsetcorner^{(|\secretspace|)}$.
		The elements of $\distsetcorner$ are known in the literature as the
		\textit{corner points} of $\errorguesspriors(\priors)$;
	\item the minimizing prior $\minprior$ of $\beta(\priors, \channel)$ has cardinality 2;
		that is, $\minprior \in \distsetcorner^{(2)}$.
\end{enumerate}

The proof for the first step comes from the observation that the
function $\beta$ is the ratio
between a concave function, $\bayesrisk$, and a function $\errorguesspriors$
that is convexly generated by $\distsetcorner$.
Lemma~\ref{lemma:min-of-ratio} (Appendix~\ref{sec:proofminbeta}) shows that the minima of this
ratio exist, and that they must come from the set of
corner points of $\errorguesspriors$ (i.e., the set $\distsetcorner$).
This determines the form of the minimizing priors.

For the second step, under the constraints given by Lemma~\ref{lemma:min-of-ratio}, the Bayes risk
$\bayesrisk(\priors, \channel)$ decreases quicker than $\errorguesspriors(\priors)$
as we increase the number of 0'es in $\priors \in \distsetcorner$.
By excluding the solution $\minprior \in \distsetcorner^{(1)}$,
which would force the denominator $\errorguesspriors = 0$,
it follows that the minimizer of $\beta(\priors, \channel)$,
$\minprior$, has exactly 2 non-zero elements; that is,
$\minprior \in \distsetcorner^{(2)}$.

\parabf{Discussion}
\autoref{thm:minbeta} has several consequences.
First, the fact that $\beta^*(\channel)$ does not depend on
a prior means that it captures the actual leakage of the channel,
excluding any prior knowledge that the adversary may have.
Conveniently, after Bayes security is computed for the channel,
one can recover the success rate of the Bayes optimal adversary for desired
levels of attacker's knowledge (\autoref{sec:MultLeak}).
Second, the fact that Bayes security represents the
\textit{risk for the two leakiest secrets} means that:
\begin{itemize}
	\item If the two leakiest secrets can be determined a priori,
		this makes the security analysis straightforward (\autoref{sec:case-study});
	\item If the two leakiest secrets cannot be determined a priori,
		one only needs $O(n^2)$ computations (instead of $O(n!)$)
		to recover them.
\end{itemize}
Finally, \autoref{thm:minbeta} suggests that Bayes security
can be interpreted as a middle way between worst-case and
average-case security metrics:
it represents the expected (i.e., average) risk for the
two most vulnerable (i.e., worst-case) secrets.
We argue that this,
paired with the fact that Bayes security is threat-specific,
favors the interpretability of this metric.
In the next sections, we prove properties about Bayes
security which make it suitable for studying complex mechanisms.

\begin{figure}[ht!]
	\centering
	\includegraphics[width=0.45\linewidth]{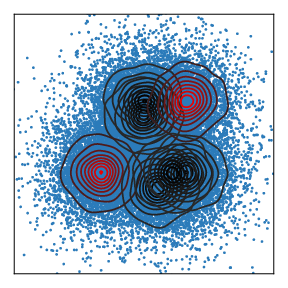}
	\caption{Posterior probability distribution for 5 secrets obfuscated
	with a two-dimensional Laplace. The Bayes security metric is
	the complement of the total variation distance between the posterior of the most distinguishable secrets (shown in red).}
	\label{fig:result-intuition}
\end{figure}

\parabf{Bayes security and total variation}
We now introduce an important result for $\beta^*$ which
helps analyzing mechanisms in practice:
Bayes security is the complement of the total
variation of the two maximally distant rows of the
channel:

\begin{restatable}{theorem}{ThmBetaStarAsDiameter}\label{thm:beta-star-as-diameter}
	For any channel $\channel$, it holds that
	\[
	\beta^*(\channel)
	~=~ 1 - \frac{1}{2}\max_{a,b\in\secretspace}\norm{\channel_a - \channel_b}_1
	~=~ 1 - \max_{a,b\in\secretspace}\TotalVar(\channel_a, \channel_b)
	~.
	\]
\end{restatable}

This result gives a clear interpretation of what $\beta^*$
represents: it measures the maximal distance
between the pairwise posterior distributions of the outputs
w.r.t. the secret inputs (\autoref{fig:result-intuition}).
Further, thanks to this result: i) it is easy to analyze mechanisms
both analytically (\autoref{sec:case-study})
and via estimation techniques (\autoref{sec:fast-minimization})
by exploiting the plethora of results surrounding the total
variation distance between distributions.

\section{The Bayes Security Metric under Composition}\label{sec:composing}

Some of the properties that made DP so popular
for studying complex algorithms are its compositionality rules:
given DP-compliant mechanisms, it is very easy to determine
the privacy of a mechanism that combines them
(e.g., by chaining them).
Further, compositionality enables studying complex threat
scenarios.
For example, while so far we have only considered an adversary who
observes the channel's output once,
it
is common that
a real adversary observes more than one channel at a time; e.g.,
observing obfuscated locations at different layers~\cite{WangG0WJS18} or
combining side channels~\cite{StandaertA08}.
Moreover, they can observe the output of two sequential channels, e.g., users' privacy-preserving interactions with a database through anonymous communication channels~\cite{ToledoDG16};
or they can observe more than one output from one channel,
e.g., by gathering several side channel measurements from
a hardware running cryptographic routines~\cite{Kocher96},
or observing more than one visit to a website through an anonymous
communication channel~\cite{wang2013website,juarez2014critical}.

In this section, we uncover compositionality rules for
Bayes security, which enable it to tackle the above examples.

\subsection{Parallel composition}
We first consider an adversary who has access
to the outputs of two channels that have as input the same secret~\cite{StandaertA08,WangG0WJS18} or where an adversary observes multiple channel outputs belonging to the same secret~\cite{Kocher96,wang2013website,juarez2014critical}.

Given two channels, $\channel_1: \secretspace\to\objectspace^1$ and $\channel_2:\secretspace\to\objectspace^2$,
their parallel composition is the channel
$\channel^1 || \channel^2: \secretspace\to\objectspace^1\times\objectspace^2$,
defined by $(\channel^1 || \channel^2)_{s,(o_1,o_2)} = \channel^1_{s,o_1} \cdot \channel^2_{s,o_2}$.

\begin{restatable}{theorem}{ThmParallelComp}
For all channels $\channel^1,\channel^2$ it holds that
\begin{equation*}
\minbeta(\channel^1 || \channel^2) \Wide\geq \minbeta(\channel^1)\cdot\minbeta(\channel^2) \,.
\end{equation*}
\end{restatable}

In layman's terms, the composition of two channels that
are respectively $\beta^*_1$-secure and $\beta^*_2$-secure
leads to a $\beta^*_1\beta^*_2$-secure channel.
This bound is tight.

Note that the security of this new channel is \textit{not
necessarily} minimized by the secrets that minimize the
composing channels $\channel^1,\channel^2$, not even
when the channel is composed with itself:

\begin{proposition}
Let $\channel$ be a channel for which $\beta$ is minimized
for secrets $(s_1, s_2)$.
Then the composition channel $\channel' := \channel || \channel$
is not necessarily minimized by secrets
$(s_1, s_2)$.
\end{proposition}

\begin{proof}
A counterexample follows.

\begin{equation*}
\channel = \begin{bmatrix}
0.9 & 0.1 & 0.0\\
0.8 & 0.2 & 0.0\\
0.5 & 0.5 & 0.0\\
0.5 & 0.1 & 0.4
\end{bmatrix}
\end{equation*}
$\minbeta(\channel) = 0.6$ is obtained for secrets
$(s_1, s_3)$;
$\minbeta(\channel || \channel) = 0.36$ is achieved
for secrets $(s_2, s_4)$.
\end{proof}

\subsection{Chaining mechanisms} %

Another typical configuration, used to strengthen
the security of the system, is to put in place a cascade of security mechanisms
(in-depth security). More formally, consider two channels
$\channel^1: \secretspace^1 \mapsto \secretspace^2$ and
$\channel^2: \secretspace^2 \mapsto \objectspace$.
Their \emph{cascade} composition is the channel
$\channel^1\channel^2$  in which the secret is input
in $\channel^1$ and this channel's output is post-processed
by $\channel^2$.

It is well understood that post-processing cannot decrease the security of a
mechanism. Therefore, $\channel^1\channel^2$ should be at least as secure ask
$\channel^1$. Indeed, based on the concavity of $\bayesrisk$, it is easy to show that
$\bayesrisk(\pi,\channel^1\channel^2) \ge \bayesrisk(\pi,\channel^1)$ for any
prior $\pi$. Consequently,
$\beta(\pi,\channel^1\channel^2) \ge \beta(\pi,\channel^1)$.

Understanding the effect of $\channel^1$ on $\channel^2$ is less
straightforward. The composition $\channel^1\channel^2$ can be seen as the
\emph{pre}-processing of $\channel^2$, which is not necessarily a safe operation.
Note that $\channel^2$ receives as input the output of $\channel^1$, which is not
necessarily the same as $\secretspace^1$.
Hence, the prior $\pi$ on the input secret in $\secretspace^1$ is meaningless for
$\channel^2$.
Remarkably, as $\beta^*$ does not depend on the prior, it allows to compare $\channel^1\channel^2$ and $\channel^2$
despite the different input spaces.
From \autoref{thm:beta-star-as-diameter} we know that $\beta^*(\channel^1\channel^2)$ is given by the
maximum $\ell_1$ distance between the rows of $\channel^1\channel^2$. The key observation is that
the rows of $\channel^1\channel^2$ are \emph{convex combinations} of the rows of $\channel^2$;
but convex combinations cannot increase distances, which brings us to the following result.

\begin{restatable}{theorem}{ThmCascadeComp}
	\label{thm:cascade-comp}
	For all channels $\channel^1,\channel^2$ it holds that
	$$\minbeta(\channel^1\channel^2) \Wide\geq \max \{ \minbeta(\channel^1), \minbeta(\channel^2)\} \,.$$
\end{restatable}

This means that \textit{neither pre-processing nor post-processing} decreases the Bayes security
provided by a mechanism.%

\section{Relation with other notions}
\label{sec:other-notions}

In the previous sections, we presented Bayes security, discussed its properties and
showed how to compute it in an efficient manner. In this section, we compare it
with three well-known security notions: cryptographic advantage,
a mainstream threat-specific metric in the security community;
DP, the paradigmatic worst-case metric;
and multiplicative Bayes vulnerability leakage,
which is closely related to $\beta$ but comes with different properties.

\subsection{Cryptographic advantage}\label{sec:Adv}

In cryptography, the advantage \advantage of an adversary \adversary is
defined assuming that there are two secrets
($|\secretspace|=2$) with a uniform prior as input to a channel
\channel. Formally (e.g. \cite{Yeom:18:CSF}):
$$\advantage(\channel,\adversary) \;\stackrel{\rm def}{=}\; 2|R^{\adversary}(\unipriors, \channel) - \nicefrac{1}{2}| \,.$$
The factor $2$ serves to scale $\advantage$
within the interval $[0,1]$.

Denoting by  $\advantage(\channel)$  the advantage of the optimal (Bayes) adversary
and considering a uniform prior $\priors=\unipriors$,
we derive:
\begin{equation}\label{eq:advbeta}
\beta(\unipriors, \channel) \; = \; 1 - 2|\bayesrisk(\unipriors, \channel)-
	\errorguesspriors(\unipriors)| = 1 - \advantage(\channel).
\end{equation}
Hence the Bayes security metric can be seen as a generalization of $1-\advantage$
for which the secret space $\secretspace$ may contain more than
two secrets the prior is not necessarily uniform.\footnote{Note that
in cryptography the advantage is usually defined for a generic
(and not necessarily optimal) adversary.}

\parabf{Bayes security as IND-CPA security} In \autoref{sec:preliminaries}, we introduced the
IND-BAY game to formalize the adversarial setting captured by
$\beta(\pi, \channel)$. When considering this game in the light
of the minimizer, $\minbeta(\channel)$, and our
main result \autoref{sec:main-result}
(\textit{$\beta$ is minimized on the two leakiest secrets}),
the IND-BAY game
becomes a version of the traditional IND-CPA cryptographic game that
we call IND-MINBAY (\autoref{fig:sec-games}, left).

First, recall that the adversary has perfect knowledge
of the prior $\priors$ and the channel $\channel$ (line 1).
Then, as opposed to the IND-BAY game, where the adversary cannot
influence the input, we allow $\adversary$ to select the secrets and provide
them to the challenger (line 2).
This is analogous to classical IND-CPA, and it allows to
capture the worst-case inputs. Then the challenger
selects one of the two secrets according to the prior $\priors$ (line 3),
and returns to the adversary an obfuscated version according to
the channel probability matrix (line 4). The adversary guesses
one of the two secrets (line 5), and wins the game if the guess
is the secret selected by the challenger. The advantage of this
adversary is equivalent to that of a CPA adversary guessing what message
was encrypted by the challenger.

This equivalence of games reinforces that the Bayes security metric
sits in the middle between average metrics (measuring the expected
risk) and worst-case metrics (measuring the worst-case risk across
the secrets).
In the next part of this section we explore this relation
further.

\begin{figure}
\begin{pchstack} [center]
	\centering
	\centering
	\procedure[linenumbering,syntaxhighlight=off]{IND-MINBAY$_\channel^\adversary$}{%
	\adversary \gets \channel,\priors\\
	\adversary \text{ selects } s_1, s_2 \in \secretspace\\
	s \overset{\priors_{1, 2}}{\gets} \{s_1, s_2\}\\
	o \overset{P(o|s)}{\gets} \objectspace\\
	s' \gets \adversary(o)\\
	\pcreturn s = s'
}
\pchspace

	\procedure[linenumbering,syntaxhighlight=off]{IND-LDP$_\channel^\adversary$}{%
		\adversary \gets \channel,\priors\\
		\adversary \text{ selects } s_1, s_2 \in \secretspace\\
		\adversary \text{ selects } o \in \objectspace\\
		s \overset{P(s|o)}{\gets} \{s_1, s_2\}\\
		s' \gets \adversary(o)\\
		\pcreturn s = s'
		}
\end{pchstack}
\caption{Security games for \minbeta (left) and LDP (right).}\label{fig:sec-games}
\end{figure}{}

\subsection{Local Differential Privacy}
\label{sec:ldp}

We investigate the relation between the privacy
guarantees induced by Bayes security (and, more in general,
$\beta$), and
those induced by DP metrics.

For a parameter $\diffp\geq 0$, we say that a mechanism is
$\diffp$-LDP (local DP)~\cite{duchi2013local} if for every $i,h,j$:
\begin{equation}
\label{eqn:LDP}
\channel_{ s_i,o_j}
\; \leq \; \exp(\diffp) \, \channel_{ s_h,o_j} \,.
\end{equation}

LDP is a worst-case metric, while recall that $\beta$
has the characteristics of an average metric.
Therefore, we expect that LDP implies a lower bound on $\beta$, but not vice versa.
The rest of this section is dedicated to analyzing this implication.

\parabf{A game for LDP}
We first illustrate the difference between the threat
model of Bayes security and the one considered by Local Differential Privacy using security
games. \autoref{fig:sec-games}, right, represents the game for local differential privacy
(IND-LDP).\footnote{We use this game for a qualitative comparison between the metrics.
	However, we observe that the parameter $\varepsilon$ of LDP can be recovered from this game by
	ensuring a uniform prior when sampling from $P(s\mid o) = \nicefrac{P(o \mid s)}{(2P(o))}$ (i.e., $P(s_1) = P(s_2) = \nicefrac{1}{2}$),
	and by evaluating the game with the following success metric:
	$\varepsilon = \ln(\nicefrac{V^*}{1-V^*})$, where $V^*=\max_s P(s \mid o)$ is the probability
	that a Bayes-optimal adversary guesses the secret correctly.
}
A first remarkable difference with respect to typical security games is that,
in addition to selecting the secrets as the IND-MINBAY game,
the adversary also chooses the observation (line 3).
This captures a worst-case in which the adversary not only
picks the most vulnerable inputs, but also the output that makes
them easier to distinguish. Upon receiving the secrets and
the observation, the challenger selects one of the secrets
according to the probability that it caused the observation
(line 4). A second difference with
respect to typical games, and IND-MINBAY, is that
the challenger \textit{does not show the chosen value to the adversary}.
Otherwise, it would be a trivial win. The adversary
guesses a secret (line 5), and wins if this is the
secret the challenger chose (line 6).
Note that, because the adversary has much greater freedom in their choices,
their chances to win are considerably greater than than in IND-MINBAY
or traditional games. Therefore, the LDP game captures a stronger
attacker than
most cryptographic games, but it is much harder to map
it to a realistic threat scenario.

\parabf{LDP induces a lower bound on Bayes security}
In general, if there are no restrictions on the channel matrix, the lowest possible value for $\beta$ is $0$;
this is achieved when the adversary
can identify the value of the secret from every observable with probability $1$.
Assuming that $|\secretspace|\geq 2$  and that  $\priors$ is not concentrated on one single secret,\footnote{If the probability mass of $\priors$ is
        concentrated on one secret, then  $G(\priors)=0$ and $\beta(\priors,\channel)$ is undefined.
	However also $\bayesrisk(\priors, \channel)=0$, and $\lim_{\priors'\rightarrow\priors}\beta(\priors',\channel)=1$. }
	and that $\secretspace$ contains at least two elements,
$\beta$ can only be zero if and only if the channel contains at most one  non-$0$ value
for each column.

If the matrix is  $\exp(\varepsilon)$-LDP, however, then the ratio between two values on the same column is at most $\exp(\varepsilon)$.
Intuitively,  under this restriction,  $\beta$  cannot be  $0$ anymore: the best case for the adversary
is when the ratio is as large as possible, i.e., when it is  \emph{exactly} $\exp(\varepsilon)$.
In particular, in a $2\times 2$ channel, we expect that the minimum $\beta$ is achieved by a matrix that has the values
  $\nicefrac{\exp(\varepsilon)}{1+\exp(\varepsilon)}$ on the diagonal, and $\nicefrac{1}{1+\exp(\varepsilon)}$ in the other positions (or vice versa).
The next theorem confirms this intuition, and extends it to the general case $n\times m$.

\begin{figure}
\begin{center}
\[
	\begin{bmatrix}
	\overmat{$\textstyle k$}{a & \cdots & a} & \!\!  \overmat{$m-k$}{b & \cdots & b}\\
	b & \cdots & b & \!\!  a & \cdots & a \\
	c & \cdots & c & \!\!  c & \cdots & c \\
	\vdots & \ddots & \vdots & \!\!  \vdots & \ddots & \vdots \\
	c & \cdots & c & \!\!  c & \cdots & c
	\end{bmatrix}
	\quad\quad\quad
	\begin{bmatrix}
	d &  e &  \; \overmat{$m-2$}{0 \; & \cdots & 0} \\
	e &  d & \; 0 \;  & \cdots & 0 \\
	1 & 0 & \; 0 \;  & \cdots & 0\\
	\vdots & \vdots & \; \vdots \; &  \ddots & \vdots \\
	1 & 0 & \; 0 \;  & \cdots & 0\\
	\end{bmatrix}
\]
\end{center}
\caption{Two examples of $n\times m$ matrices $\minchannel$ which achieve minimum $\beta$ value $\minbeta(\minchannel)=\frac{2}{1+\exp(\varepsilon)}$.
In the first matrix: $a =  \frac{\exp(\varepsilon)}{k(1+\exp(\varepsilon))}$, $b =   \frac{1}{(m-k)(1+\exp(\varepsilon))}$ and $c=    \frac{1}{m}$.
In the second matrix:   $d  =  \frac{\exp(\varepsilon)}{ 1+\exp(\varepsilon)}$ and $e  =  \frac{1}{ 1+\exp(\varepsilon)}$.
 }\label{fig:min channels}
\end{figure}

\begin{restatable}{theorem}{ThmBoundOnBayes}
\label{theo:BoundOnBayes}\qquad\qquad
\begin{enumerate}
\item \label{LDP1}If  \channel is $\varepsilon$-LDP, then for every $\priors$ we have  $\beta(\priors,\channel)\geq \frac{2}{1+\exp(\varepsilon)}$.\\
\item \label{LDP2}For every $n,m\geq 2$ there exists a $n\times m$ $\varepsilon$-LDP channel
$\minchannel$ such that $\minbeta(\minchannel) = \frac{2}{1+\exp(\varepsilon)}$.
Examples of such $\minchannel$'s are illustrated in  \autoref{fig:min channels}.
\end{enumerate}
\end{restatable}

\parabf{Bayes security does not induce a lower bound on LDP}
\autoref{theo:BoundOnBayes} shows that $\varepsilon$-LDP  induces a bound on Bayes security, and
that we can express a strict bound that depends only on $\varepsilon$.
The other direction does not hold. The main reason is that if
a column contains both a $0$ and a positive element,
then  $\varepsilon$-LDP  cannot hold, independently from the value of $\beta$.

\subsection{Approximate Differential Privacy}
One may consider  $(\varepsilon,\delta)$-LDP~\cite{dwork2006our}. This is   a variant of LDP in which small violations to
\autoref{eqn:LDP} are tolerated. Precisely,
a mechanism is
$(\varepsilon,\delta)$-LDP if for every $s_i,s_j\in \secretspace$ and $O\subseteq \objectspace$:
	\begin{equation}
	\label{eqn:edLDP}
\sum_{o\in O}  \left(\channel_{s_i,o} - \exp(\diffp)\,\channel_{s_j,o} \right)\; \leq \; \delta \,.
\end{equation}

With $(\varepsilon,\delta)$-LDP, a column may contain $0$ and non-$0$ values, as long as the latter are smaller than $\delta$.
Similarly to pure DP, approximate DP is threat-agnostic; this makes it harder to match $(\varepsilon,\delta)$ values to the risk of an attack occurring.

Surprisingly, we observe a direct relation between Bayes security and the special case $(0, \delta)$-DP:

\begin{proposition}
	Let $\channel$ be a $\beta^*$-secure channel.
	Then it is also $(0,\delta)$-LDP, with $\delta = 1-\beta^*$.
\end{proposition}

This comes from the fact that, for $\diffp = 0$, the LHS of \autoref{eqn:edLDP} becomes $\sum_{o\in O}  \left(\channel_{s_i,o} -  \channel_{s_j,o} \right)$, which is maximized for
$O^*=\{o\in \objectspace\mid \channel_{s_i,o}> \channel_{s_j,o}\}$.
Observe that $\sum_{o\in O^*}  \left(\channel_{s_i,o} -  \channel_{s_j,o} \right)$ corresponds to the total
variation between $\channel_{s_i,o}$ and $\channel_{s_j,o}$.
Applying the equivalence between $\beta^*$ and total variation~(\autoref{thm:beta-star-as-diameter})
concludes the argument.

The special case of $(0,\delta)$-LDP mechanisms is not commonly studied. Intuitively, it corresponds to a mechanism that is completely vulnerable, but only with probability $\delta$.
We hope that the direct correspondence between
$\beta^*$, $\delta$, and the cryptographic advantage can give further insights in the decision of the parameter choices for approximate DP.

\subsection{Differential privacy~\cite{Dwork:06:TCC}}
Differential privacy is similar to LDP, except that it involves the notion of adjacent databases.
Two databases $x,x'$ are adjacent, denoted as $x\sim x'$, if $x$ is obtained from $x'$ by removing or adding one record.

The definition of $\varepsilon$-differential-privacy ($\varepsilon$-DP), in the discrete case, is as follows.
A mechanism $\cal K$ is  $\varepsilon$-DP  if for every $x,x'$ such that  $x\sim x'$, and every $y$, we have
$$P({\cal K}(x) = y ) \;\leq \; \exp(\varepsilon)\, P({\cal K}(x') = y ) . $$

A relation between the Bayes security and DP follows from an analogous result in \cite{Alvim:15:JCS}
for the multiplicative Bayes leakage $\mulleakage(\pi,{\cal K})$, and the correspondence between the latter and the Bayes security (cfr. Section~\ref{sec:MultLeak}), which is given by
\begin{equation}\label{eq:BS-MultLeak}
\beta(\priors,\channel) = \frac{1-(\max_s\pi_s) \, \mulleakage(\pi,\channel)}{1 - \max_s\pi_s}.
\end{equation}

The following result, proven by Alvim et al.~\cite{Alvim:15:JCS},  states that
$\varepsilon$-DP induces a  bound on the  multiplicative vulnerability leakage, where the set of secrets are all the possible databases.
The theorem is given for the bounded DP case, where we assume that the number of records present in the database is at most a certain number $n$, and
that  the set of values for the records includes a special value $\bot$ representing the absence of the record.
The adjacency relation is  modified accordingly: $x\sim x'$ means that $x$ and $x'$ differ for the value of exactly one record.
We also assume that the cardinality  $\mathsf{v}$ of the set of values is finite. Hence also the number of secrets (i.e., the possible databases) is finite.

\begin{theorem}[From \cite{Alvim:15:JCS}, Theorem 15]\label{theo:boundMultLeak}
If  ${\cal K}$ is $\varepsilon$-DP, then, for every $\pi$, $\mulleakage(\pi,\channel)$  is bounded from above as
\[
\mulleakage(\pi,\channel) \; \leq \;
	\left(\frac{\mathsf{v}\;\exp({\varepsilon})}{\mathsf{v} - 1 + \exp({\varepsilon})}\right)^n.
\]
and this bound is tight when $\pi$ is uniform.
\end{theorem}

From \autoref{theo:boundMultLeak} and \autoref{eq:BS-MultLeak} we immediately obtain a bound also for the Bayes security:

\begin{corollary}\label{cor:boundBS}
If  ${\cal K}$ is $\varepsilon$-DP, then, for every $\pi$, $\beta(\pi,\channel)$   is bounded from below as
\[
\beta(\pi,\channel) \; \geq \; \frac{1-(\max_s\pi_s) \left( \frac{\mathsf{v}\;\exp({\varepsilon})}{\mathsf{v} - 1 + \exp({\varepsilon})}\right)^n}{1 - \max_s\pi_s}.
\]
and this bound is tight when $\pi$ is the uniform distribution $\unipriors$ which assigns $\nicefrac{1}{\mathsf{v}^n}$ to every database, in which case it case it can be rewritten as
\[
\beta(\pi,\channel) \; \geq \; \frac{\mathsf{v}^n-\left( \frac{\mathsf{v}\;\exp({\varepsilon})}{\mathsf{v} - 1 + \exp({\varepsilon})}\right)^n}{\mathsf{v}^n - 1}.
\]
\end{corollary}

 Alvim et al.~\cite{Alvim:15:JCS} show that the reverse of \autoref{theo:boundMultLeak} does not hold, and as a consequence the reverse of \autoref{cor:boundBS} does not hold either.
 The reason is analogous to the case of LDP: a $0$ in a position of a non-$0$-column implies that the mechanism cannot be DP, independently from the value of $\beta$.

\parabf{Membership inference}
In \autoref{cor:boundBS} the secrets are  the whole databases. Often, however, in DP we assume that the attacker is not interested at discovering the whole database, but only whether
a certain record belongs to the database or not.
We can model this case by isolating a generic pair of adjacent databases $x$ and $x'$, and then restricting the space of secrets to be just $\{x,x'\}$.
On this space, the mechanism can  be represented by a stochastic channel $\channel^{\{x,x'\}}$ that has only the two inputs $x$ and $x'$, and as  outputs the (obfuscated) answers to the query.
It is immediate to see that ${\cal K}$ is $\varepsilon$-DP iff $\channel^{\{x,x'\}}$ is $\varepsilon$-LDP for any pair of adjacent databases $x$ and $x'$.
Hence, the relations we proved between
Bayes security and LDP hold also for DP.
In particular, the following is an immediate consequence of \autoref{theo:BoundOnBayes}.
\begin{corollary}\label{cor:BoundOnBayes}
 If  ${\cal K}$ is $\varepsilon$-DP, then for every pair of adjacent databases $x$ and $x'$ and every
$\priors$ we have
\[\beta(\priors,\channel^{\{x,x'\}})\geq \frac{2}{1+\exp(\varepsilon)}.\]
and this bound is strict.
\end{corollary}

A similar investigation was done by Yeom et al.~\cite{Yeom:18:CSF}.
They studied the privacy of $\channel^{\{x,x'\}}$ in terms of the advantage,
defined in the context of membership inference attacks (MIA).
The authors established that, if a mechanism is  $\varepsilon$-DP,
then the following  lower bound for holds for $\advantage(\channel^{\{x,x'\}})$,
for any adjacent databases $x$ and $x'$:
\begin{equation}\label{eq:adv-DP}
\advantage(\channel^{\{x,x'\}}) \leq  \exp(\varepsilon)-1
\end{equation}
By using the relation between the advantage and the Bayes security
metric (\autoref{eq:advbeta}), we derive the following bound:
\begin{equation}\label{eq:boundfromadv}
\beta(\unipriors,\channel^{\{x,x'\}})\geq  2-\exp(\varepsilon) \,.
\end{equation}
where $\unipriors$ is the uniform distribution.

\begin{figure}
	\centering
	\includegraphics[width=0.4\linewidth]{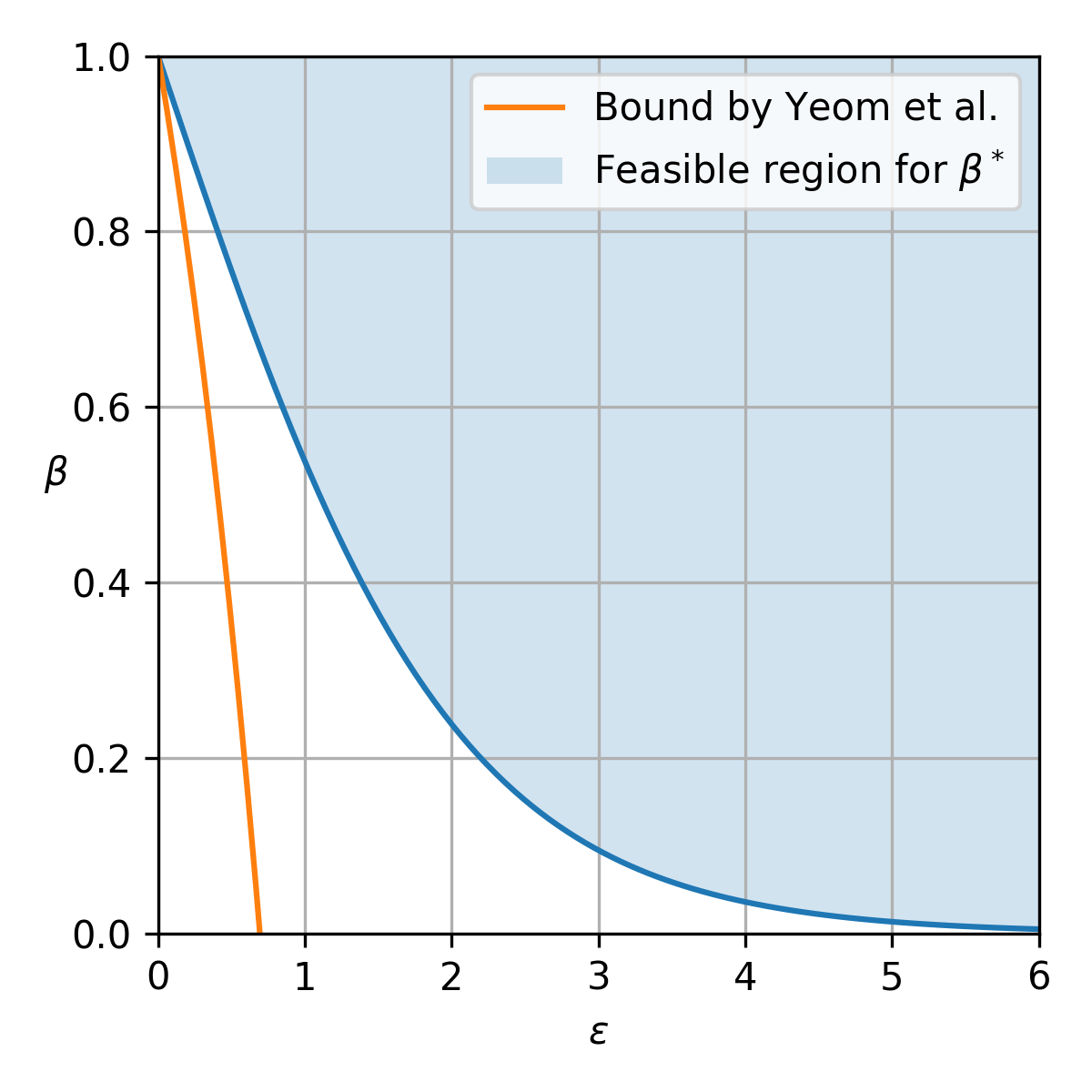}
	\caption{The blue line illustrates the lower bound of
		$\varepsilon$-DP on $\beta$ expressed by Corollary~\ref{cor:BoundOnBayes}).
		The orange line represents the lower bound on  $\beta$ derived from the
		one proved in  Yeom et al.~\cite{yeom2018privacy} for the advantage of a membership inference adversary.}
	\label{img:beta-dp}
\end{figure}

By exploiting the equivalence between Bayes security and the advantage,
we conclude that the bound by Yeom et al.~\cite{Yeom:18:CSF} is loose;
from  \autoref{cor:BoundOnBayes} and \autoref{eq:advbeta},
we derive the following (strict) bound for the advantage of an $\varepsilon$-DP mechanism:
\[
\advantage(\channel^{\{x,x'\}}) \leq  \frac{\exp(\varepsilon)-1}{\exp(\varepsilon)+1}.
\]
which is much tighter than their bound
(see \autoref{img:beta-dp}).

Concurrent work by Humphries et al.~\cite{humphries2020differentially} proved a similar bound for
$(\varepsilon, \delta)$-DP in the context of membership inference
attacks.
Their bound is more general than ours, since it captures
approximate DP; however, we prove tightness for our bound.
Whether tightness can be proven for the bound by Humphries et al. is
to our knowledge an open problem.

\subsection{Leakage notions from Quantitative Information Flow}\label{sec:MultLeak}

We discuss multiplicative risk leakage ($\beta$) and its minimizer
($\beta^*$) from the point of view of Quantitative Information Flow (QIF),
and compare it with similar metrics stemming from the field.
QIF measures the information leakage of a
system by comparing its vulnerability \emph{before} and \emph{after}
observing its output. It starts with a \emph{vulnerability} metric
$V(\pi)$, expressing how vulnerable the system is when the adversary has
knowledge $\pi$ about the secret.
The \emph{posterior} vulnerability is defined as
$V(\pi, \channel) = \sum_o p(o) V(\delta^o)$, where $\delta^o$ is the posterior
distribution on $\secretspace$ produced by the observation $o$; intuitively, it expresses how vulnerable
the system is, on average,
\emph{after observing} the system's output.
\emph{Leakage} is defined by comparing the two,
either \emph{multiplicatively} or \emph{additively}:
\[
	\mulleakage(\priors, \channel) = \frac{V(\pi,\channel)}{V(\pi)} ~,
	\quad
	\addleakage(\priors, \channel) = V(\pi,\channel) - V(\pi) ~.
\]

One of the most widely used vulnerability metrics is \emph{Bayes vulnerability} \cite{smith2009foundations},
defined as $V(\pi) = \max_s \pi_s = 1 - \errorguesspriors(\priors)$;
it expresses the
adversary's probability of guessing the secret correctly in one try.\footnote{%
The tightly connected notion of \emph{min-entropy}, defined as $\log V(\pi)$, is used by many authors
instead of Bayes vulnerability.
}
For the posterior version,
it holds that $V(\pi,\channel) = 1 - \bayesrisk(\pi,\channel)$.
The multiplicative risk leakage follows the same core idea: $\errorguesspriors(\priors)$ can be
thought of as a \emph{prior} version of $\bayesrisk$: indeed, it holds that
$\bayesrisk(\pi, \channel) = \sum_o p(o) \errorguesspriors(\delta^o)$ where $\delta^o$ are the
posteriors of the channel. Hence, $\beta$ can be considered to be a variant of
multiplicative vulnerability leakage, using Bayes risk instead of Bayes vulnerability.

Since the two are closely related, one would expect to be
able to directly translate results about $\mulleakage(\priors, \channel)$ to similar results
on $\beta(\pi,\channel)$. This would be the case for \emph{additive leakage}, since
$V(\pi,\channel) - V(\pi) = \errorguesspriors(\pi) - \bayesrisk(\pi,\channel)$, but in the multiplicative
case, the ``one minus'' in both sides of the fraction completely changes the behavior of the
function.

\begin{table*}[h!]
	\caption{Security metrics comparison: Local Differential Privacy (LDP),
		Multiplicative leakage capacity, and Bayes Security (\minbeta).
		Note that the cryptographic advantage is a special case of $\minbeta$, and
		therefore not included in this table. ``Consistent Black-box Estimation'' refers to the existence of a statistically consistent estimator for the security metric (e.g., \cite{cherubin2019fbleau}).}
	\label{tab:comparison}
	\centering
	\def\arraystretch{1.2}
	\begin{tabularx}{\linewidth}{llXXX}
		& \textbf{Property} & \textbf{LDP} ($\varepsilon$) & $\mathbf{\mulcapacity}$ & $\mathbf{\minbeta}$\\
		\cline{2-5}
		& \textbf{Range} & $[0, \infty)$ & $[1, n]$ & $[0, 1]$\\
		&  & Smaller is more secure & Smaller is more secure & Larger is more secure\\
		\cline{2-5}
		& \textbf{Attacks} & Any attack & Concrete attack & Concrete attack\\
		& \textbf{Security Guarantee} &  Worst case for 2 leakiest secrets & Expected risk among all secrets & Expected risk for 2 leakiest secrets\\
		& \textbf{Qualitative Intuition} & Bounds probability of ever distinguishing two secrets %
			  & Bounds the probability of guessing among all secrets & (Complement of) the advantage of an attacker in guessing the secret w.r.t. the random baseline\\
		& \textbf{Quantitative Intuition} & None for general case. Can be defined w.r.t. mechanism &
			$\mulcapacity = k$ means the adversary is $k$ times more likely to guess the secret &
			$1-\nicefrac{1}{2}\minbeta$ is the probability that the adversary guesses the secret correctly\\ %
		\cline{2-5}
		& \textbf{Composable} & \cmark & \cmark & \cmark\\
		& \textbf{Consistent Black-box Estimation} & \xmark & \cmark & \cmark\\
		& \textbf{Prior-agnostic} & \cmark & \cmark & \cmark\\
	\end{tabularx}
\end{table*}

\parabf{Capacity vs \ensuremath{\minbeta}}
One should first note that, while $\beta$
takes lower values to indicate a worse level of security,
$\mulleakage$ takes higher values.
In both cases, a natural question is to find the prior $\pi$
that provides the worst level of security; in the
case of leakage, its maximum value is known as \emph{channel capacity},
denoted by $\mulcapacity(\channel) = \max_\pi\mulleakage(\pi, \channel)$.

$\mulcapacity(\channel)$ is given
by the \emph{uniform prior} \cite{braun2009quantitative}, and
$\mulcapacity(\channel) = \sum_o \max_s \channel_{s,o}$.
Our main result (\autoref{thm:minbeta}) shows that $\beta(\pi, \channel)$ is minimized on a uniform prior over 2 secrets.
Hence, despite the similarity between Bayes vulnerability and Bayes risk,
the corresponding leakage and security metrics behave very differently.
Note that this difference makes $\mulcapacity(\channel)$ easier to compute
for arbitrary channels; it
is linear on both $|\secretspace|$ and $|\objectspace|$, while $\minbeta$
is quadratic on $|\secretspace|$.
We discuss fast ways for computing (or estimating) Bayes security
in \autoref{sec:fast-minimization}.

\parabf{Channel composition}
Despite their difference w.r.t. the prior that realizes each notion,
$\mulcapacity$ and $\minbeta$ behave similarly w.r.t. parallel and cascade composition.
It was shown that~\cite{espinoza2011min}:
\begin{align*}
	\mulcapacity(\channel^1 || \channel^2) &\Wide\leq \mulcapacity(\channel^1)\cdot\mulcapacity(\channel^2) ~,
		&\text{and}\\
	\mulcapacity(\channel^1\channel^2) &\Wide\leq \min \{ \mulcapacity(\channel^1), \mulcapacity(\channel^2)\} ~.
\end{align*}
The same bounds are given in \autoref{sec:composing} for $\minbeta$. Note, however,
that the proofs for $\minbeta$ are completely different and cannot be directly obtained from those
of $\mulcapacity$.

\parabf{Bounds on Bayes risk}
The goal of a security analyst is to quantify
how much information is leaked by a mechanism in the \emph{worst case}.
This is captured by both $\mulcapacity$ and $\minbeta$ which
focus on the prior that produces the highest leakage instead
of the true prior.
The user, however, is mostly interested in the actual threat
that he is facing: how likely it is for the adversary to guess his secret,
given a particular prior $\pi$ that captures the user's behavior.
In this sense, the Bayes risk, $\bayesrisk(\priors, \channel)$, has a clear operational interpretation for the user.

Fortunately, having computed either $\mulcapacity$ or $\minbeta$, we can
obtain direct bounds for the prior $\pi$ of interest:
\begin{align*}
	\bayesrisk(\priors, \channel)
	&\Wide\ge \minbeta(\channel) \cdot \errorguesspriors(\pi)
		~, &\text{and} \\
	\bayesrisk(\priors, \channel)
	&\Wide\ge 1 - \mulcapacity(\channel) \cdot V(\pi)
	~.
\end{align*}

The goodness of either bound depends on the application.
Intuitively, how good the bound is depends on how close $\pi$ is to the one
achieving  $\mulcapacity$ or $\minbeta$.
Concretely, since the former implies uniform priors, and the latter
a vector with only 2 non-empty (uniform) entries, the tightness of these
bounds depends on the sparsity of the real prior vector $\pi$.

We study empirically how tight these bounds are.
We consider two channels with $|\secretspace| = 10$ inputs and
$|\objectspace| = 1K$ outputs.
The first channel (hereby referred to as the \textit{random} channel) is obtained by sampling at random its conditional
probability distribution $P(o \mid s)$;
the second one (\textit{geometric} channel) has a geometric
distribution as used by Cherubin et al.~\cite{cherubin2019fbleau}, with
a noise parameter $\nu = 0.1$.
To evaluate the effect of sparsity, we set a sparsity level to
$\sigma \in {0, 1, ..., n-2}$, and we sample a prior that is $\sigma$-sparse
uniformly at random.
We compute the values of $\mulleakage$ and $\beta$, and measure their
absolute distance respectively from $\mulcapacity$ and $\minbeta$.
The experiment is repeated $1K$ times for each sparsity level.

\begin{figure}[ht!]
	\centering
	\includegraphics[width=0.445\linewidth]{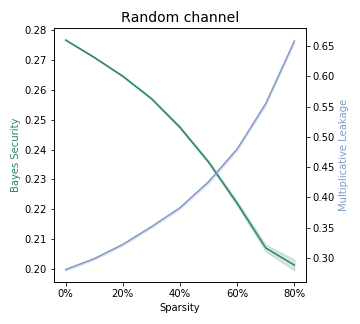}
	\includegraphics[width=0.445\linewidth]{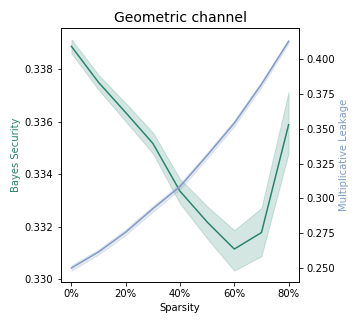}
	\caption{Tightness of the bounds on Bayes security and multiplicative leakage
		with respect to sparsity. Note that, because the two metrics have
		different scales, these plots are useful to compare their behavior,
		and not their actual values.}
	\label{img:ml-vs-beta-bound}
\end{figure}

\autoref{img:ml-vs-beta-bound} shows the results.
As expected, the multiplicative vulnerability leakage bound is tighter for vectors
that are less sparse, and the Bayes security one for higher
sparsity levels.
However, we observe that the Bayes security bound is loose
for high values of sparsity in the case of the geometric channel,
but not for the random one.
The reason is that if the real prior has maximum sparsity
(i.e., only 2 non-zero entries), then it is more
likely that the secrets on which $\beta$ is minimized are
not the same 2 secrets on which the prior is not empty.

As a consequence of this analysis,
we suggest $\mulcapacity$ is better suited to analyze deterministic mechanisms
with a large number of secrets distributed close to uniformly
(see \cite{smith2009foundations}). For deterministic
programs, $\minbeta$ is always $0$, unless the program is non-interfering; the bound
obtained by $\minbeta$ is then trivial, while  $\mulcapacity$ can provide meaningful bounds
if the number of outputs is limited.
On the other hand, $\minbeta$ is advantageous if $\pi$ is very sparse (e.g., website fingerprinting,
where the user may be visiting only a small number of websites):
since $\pi$ is very different
than the uniform one, and more similar to the one achieving $\minbeta$, the latter
will provide much better bounds.
We discuss application examples for Bayes security in \autoref{sec:discussion}.

\parabf{Miracle theorem}
Both Bayes vulnerability and Bayes risk can be thought of
as instantiations of a general family of metrics
parameterized by a \emph{gain function} $g$ (for vulnerability) or a
\emph{loss function $\ell$} (for risk). For generic
choices of $g$ and $\ell$, we can define $g$-leakage
$\mulleakage_g(\pi, \channel)$
and the corresponding security notion $\beta_\ell(\pi,\channel)$, in
a natural way (detailed in Appendix~\ref{appendix:gain-loss}).

A result by Alvim et al.~\cite{alvim2012measuring}, known as ``miracle'' due to its
arguably surprising nature, states that
\[
	\mulleakage_g(\pi,\channel) \Wide\le \mulcapacity(\channel)
	~,
\]
for all priors $\pi$ and all \emph{non-negative} gain functions $g$.
This gives a direct bound for a very general family of leakage metrics.

For $\beta_\ell$, however, we know that a corresponding result does not hold in general,
even if we restrict to the family of $[0,1]$ loss functions. Identifying
families of loss functions that provide similar bounds is left as future work.

\renewcommand{\algorithmicrequire}{\textbf{Input:}}
\renewcommand{\algorithmicensure}{\textbf{Output:}}
\renewcommand{\algorithmicrequire}{\textbf{Input:}}
\renewcommand{\algorithmicensure}{\textbf{Output:}}
\newcommand{\M}{\mathcal{M}}  %
\newcommand{\R}{\mathbb{R}}
\newcommand{\D}{d}  %
\newcommand{\Domain}{\mathcal{D}}
\newcommand{\Range}{\mathcal{R}}
\newcommand{\aux}{\textsf{aux}}
\newcommand{\E}{\mathbb{E}}
\newcommand{\leader}[1]{\medskip\noindent\textbf{#1}}
\newcommand{\eqr}[1]{Eq.~\eqref{#1}}
\newcommand{\biburlhref}[2]{\href{#1}{#2}}
\newcommand{\btheta}{{\mathbb \theta}}
\newcommand{\calL}{\ensuremath{\mathcal L}}
\newcommand{\calN}{\ensuremath{\mathcal N}}
\newcommand{\g}{\ensuremath{\mathbf g}}
\newcommand{\Id}{\ensuremath{\mathbf I}}
\newcommand{\doh}[2]{\frac{\partial #1}{\partial #2}}
\newcommand{\eps}{\varepsilon}
\newcommand{\x}{\ensuremath{\mathbf x}}

\section{Case studies: Bayes security of well-known mechanisms}
\label{sec:case-study}

\begin{figure*}
	\centering
	\begin{subfigure}[b]{0.23\textwidth}
		\centering
		\includegraphics[width=\textwidth]{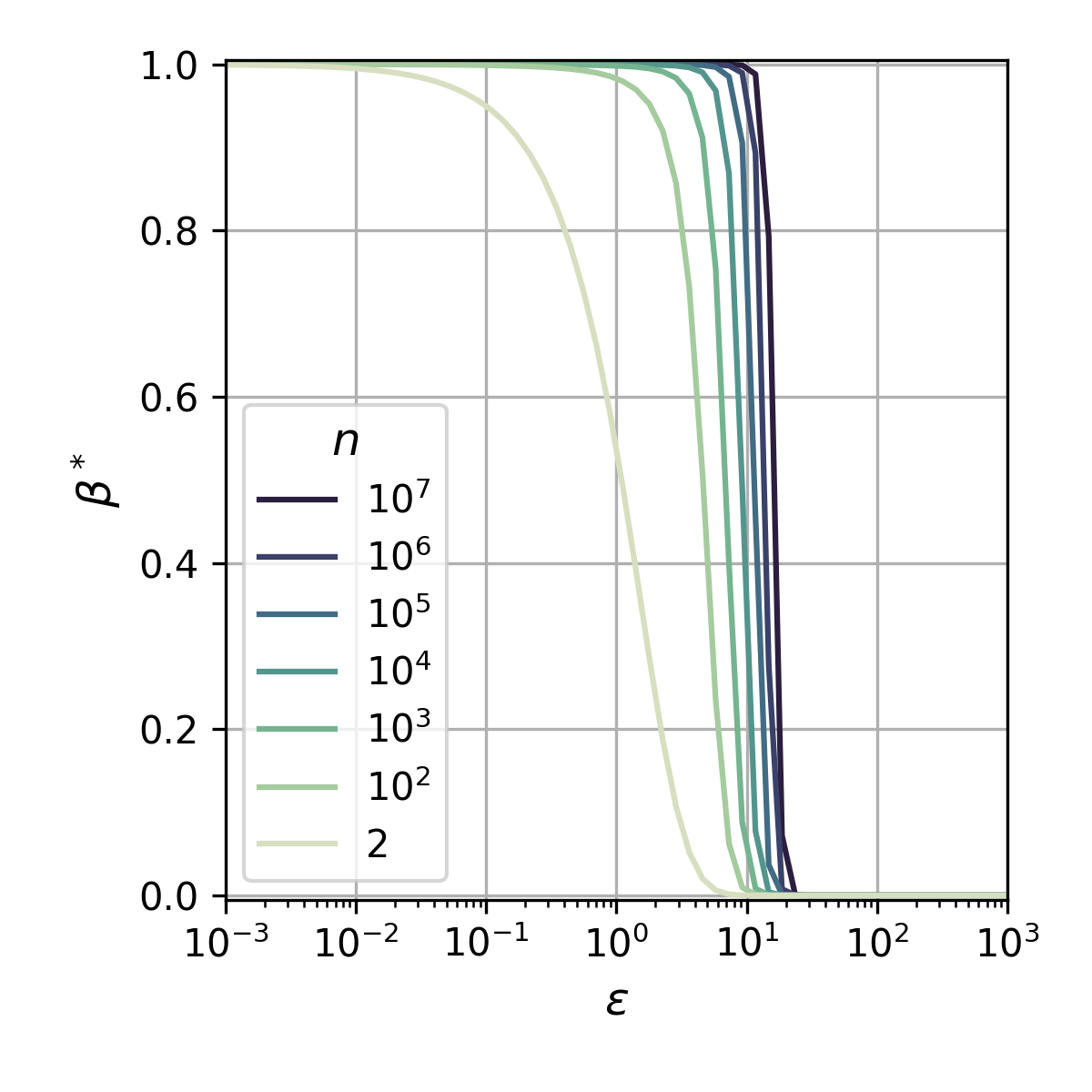}
		\caption{Randomized Response}
	\end{subfigure}
	\begin{subfigure}[b]{0.23\textwidth}
		\centering
		\includegraphics[width=\textwidth]{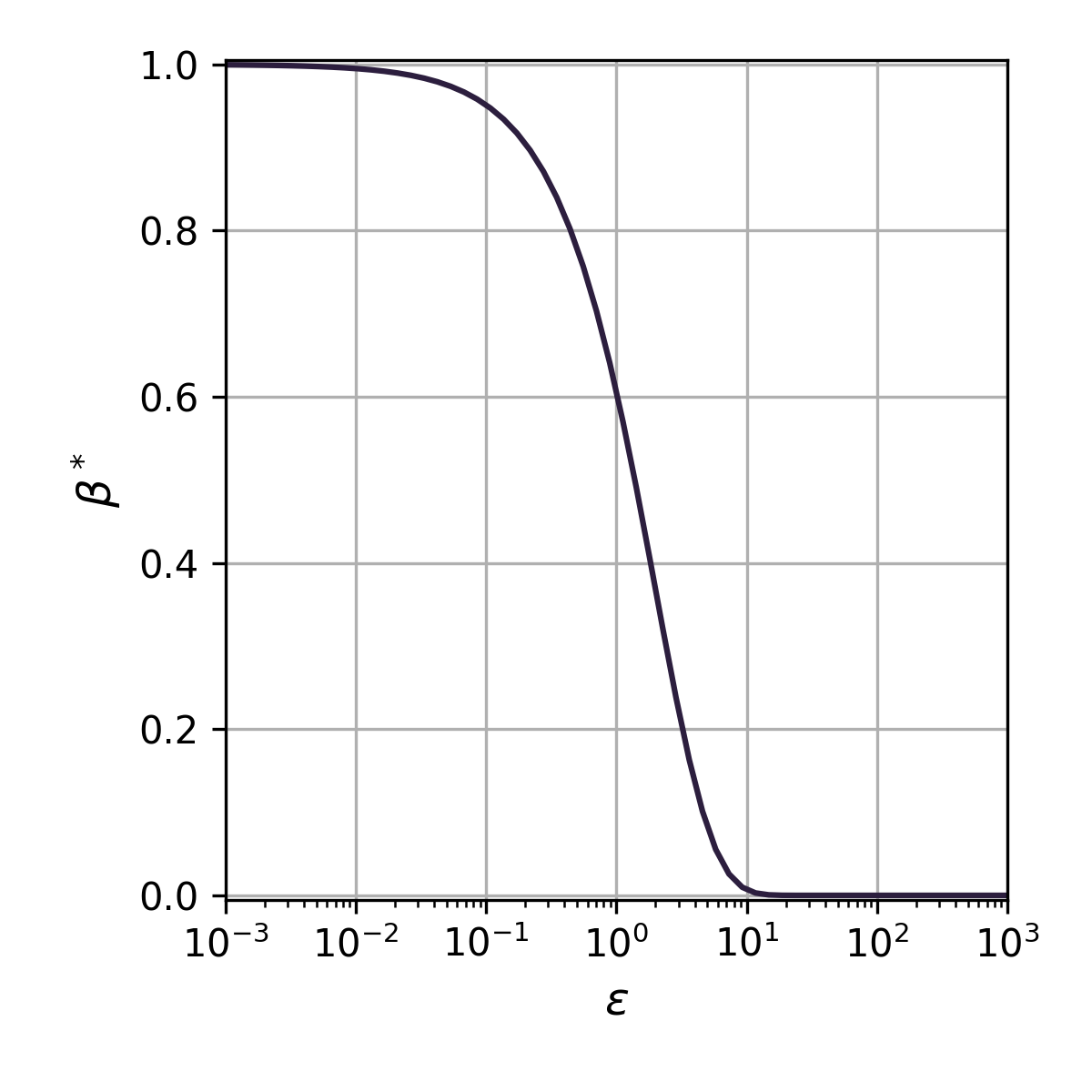}
		\caption{Laplace mechanism}
	\end{subfigure}
	\begin{subfigure}[b]{0.23\textwidth}
		\centering
		\includegraphics[width=\textwidth]{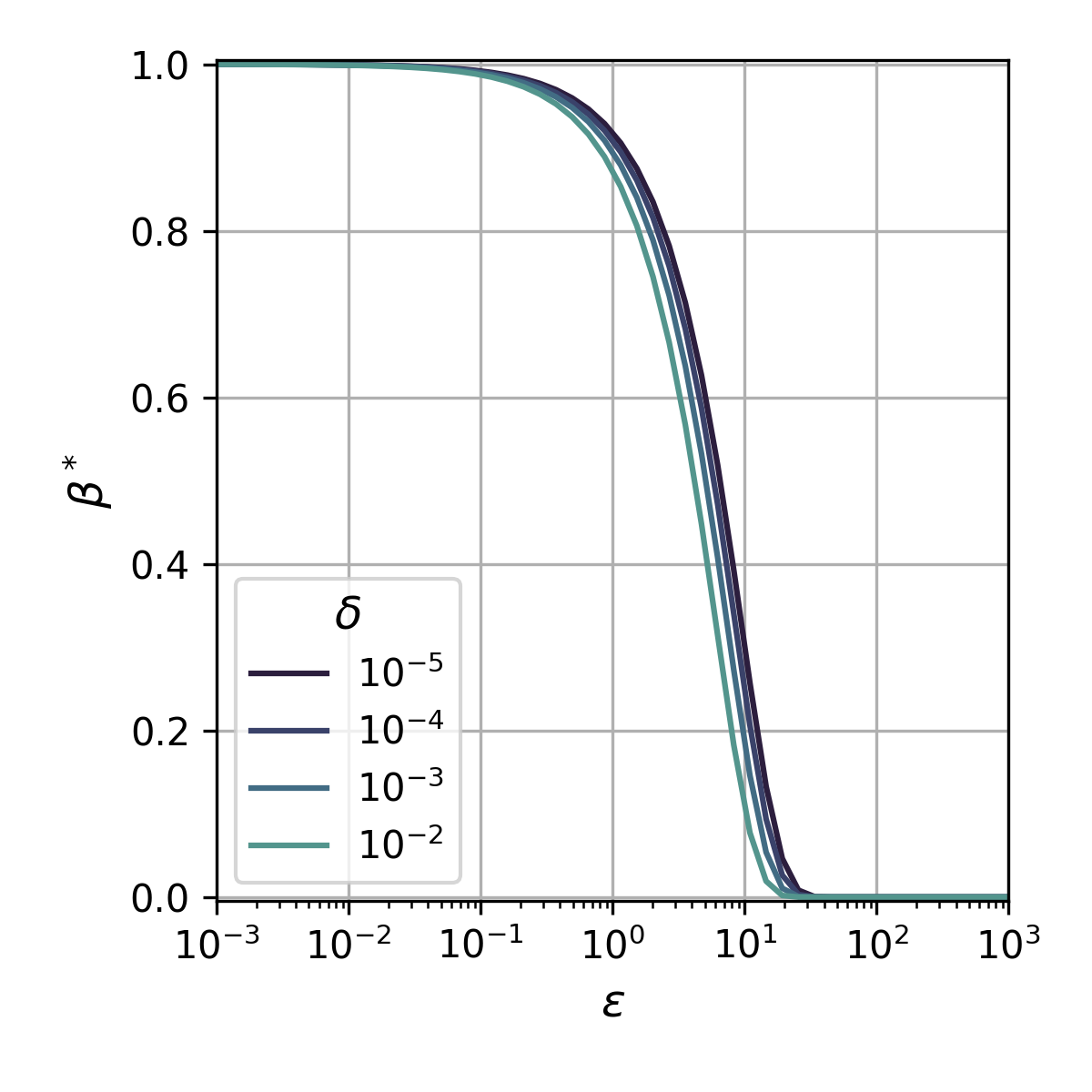}
		\caption{Gaussian mechanism}
	\end{subfigure}
	\caption{Relation between Bayes security and DP for
		various mechanisms;
$(\varepsilon, \delta)$-DP is used for the Gaussian mechanism.}
	\label{fig:case-study}
\end{figure*}

We now exploit the various properties we proved about
Bayes security to study well-known mechanisms:
Randomized Response, and the Gaussian and Laplace mechanisms.
These are often used as building blocks for more complex ones.

\subsection{Randomized Response}

Randomized Response (RR) is a simple
obfuscation protocol that guarantees $\varepsilon$-LDP.
It randomly assigns a data record to a new data record
from the same range. Assuming $\secretspace = \objectspace$,
RR is represented as the following channel
matrix, $\mathcal{R}: \secretspace \mapsto \objectspace$:
$$
\mathcal{R}_{s,o} = P(o|s) \eqdef \begin{cases}
	\frac{exp(\varepsilon)}{n + \exp(\varepsilon) - 1} \text{ if } o = s\\
	\frac{1}{n + \exp(\varepsilon) - 1} \text{ otherwise}\,.
\end{cases}
$$

We can derive $\minbeta$ easily for RR because it obfuscates each secret according to the
same distribution. For any two secrets $s_i$ and $s_j$, the
rows of the channel matrix are identical except for positions $i$ and $j$,
where their values are inverted; therefore, the Bayes security
is the same between any two secrets, and thus all are equally vulnerable.
Using the results in~\autoref{sec:efficient-estimation}, we just need to look
at any two rows, e.g., the first two.
Let $\mathcal{R}_{ab}$ indicate the sub-channel matrix
containing only the first two rows, and let
$\unipriors = \nicefrac{1}{2}$.
The corresponding Bayes risk is
$\bayesrisk(\unipriors, \mathcal{R}_{ab})=\frac{n}{2(\exp(\varepsilon)+n-1)}$;
hence the Bayes security is:
\begin{equation} \label{eq:minbetaRR}
	\minbeta(\mathcal{R}) = \frac{n}{\exp(\varepsilon) + n - 1} \,,
\end{equation}
where $n$ is the number of secrets and observables.

\parabf{Discussion}
\autoref{eq:minbetaRR} captures the risk that an optimal
adversary can distinguish between any two data records (secrets)
from the RR output; by \autoref{thm:minbeta},
these are the easiest two records to distinguish, and
it implies Bayes security for any other subset of
the secret space.

We use this equation to relate Bayes security and $\varepsilon$
w.r.t. the number of data records (secrets) in \autoref{fig:case-study} (a).
We observe that the number of data records
is essential for security.
For a rather loose DP parameter of $\varepsilon=10$,
having a dataset of $1M$ gives $\beta^*\approx0.978$;
assuming the two data records have the same prior,
the probability that the adversary guesses correctly
is $0.511$.
With the same $\varepsilon$, having a dataset of
$10M$ ensures a practically perfect Bayes security
of $0.998$ (Bayes vulnerability for a uniform prior:
$0.501$).
Overall, this shows that, if we are interested in a
specific threat model(s), then a threat-specific metric such
as Bayes security can reassure us on the security of the
mechanism even when DP suggests it is not.

In the appendix (\autoref{appendix:rr}), we include an empirical study of RR
for the \texttt{Census1990} dataset.
We observe that, in this specific case, a good utility (95\%)
is only achieved for a rather large $\varepsilon=3.3$;
in principle, one would disregard the mechanism to be unsafe
in this particular instance.
Yet, Bayes security ($\minbeta = 0.99999$ for $\varepsilon=3.3$,
and $\minbeta=0.99995$ for $\varepsilon=4.8$) reassures us on its security
within this threat model.

\subsection{Laplace mechanism}

For parameter a $\lambda$, we define the Laplace mechanism
as $\mathcal{L}: s \mapsto s + \Lambda(0, \lambda)$,
where $\Lambda(\mu, \lambda)$
is a $\mu$-centered Laplace distribution with scale $\lambda$.

\begin{restatable}{proposition}{ThmLaplace}
	\label{thm:beta-Laplace}
	$\mathcal{L}$ is $\beta^*$-secure with
	$$\beta^* %
	= \exp\left(-\max_{s_i, s_j \in \secretspace}\frac{|s_i-s_j|}{2\lambda}  \right)\,,$$
\end{restatable}

\parabf{Discussion}
We can use this analysis to compare $\beta^*$ with DP.
Let $f: \mathcal{D} \mapsto \mathbb{R}$ be a real-valued
function with sensitivity
$\Delta f \eqdef \max_{x, y \in \mathcal{D}} f(x)-f(y)$.
The mechanism $f(x) + \mathcal{L}(0, \lambda)$ with scale parameter
$\lambda = \frac{\Delta f}{\varepsilon}$ is $\varepsilon$-DP.
By using the last result, the
Bayes security of this mechanism is
$\beta^* = \exp(-\frac{\varepsilon}{2})$.

Now, suppose we care about the probability of an adversary
at distinguishing the two maximally distant points
$s_1, s_2 = \argmax_{x, y \in \mathcal{D}} f(x)-f(y)$.
For a relatively strong DP level of $\varepsilon=0.1$,
we get $\beta^* \approx 0.95$;
This implies a non-negligible advantage for the adversary;
e.g., assuming the two points have identical prior $\nicefrac{1}{2}$,
the probability that the optimal adversary distinguishes
between them is roughly $0.525$.
\autoref{fig:case-study} (b) shows the overall behavior.

\subsection{Gaussian mechanism}

For parameter $\sigma$, the Gaussian mechanism adds noise
to a secret $s$ from a Gaussian distribution:
$\mathcal{G}: s \mapsto s + \mathcal{N}(0, \sigma^2)$.

\begin{restatable}{proposition}{ThmGaussian}
	\label{thm:beta-gaussians}
	$\mathcal{G}$ is $\beta^*$-secure with
	$\beta^* = 1-(\Phi(\alpha)-\Phi(-\alpha))$,
	where $\Phi$ is the CDF of $\mathcal{N}(0,1)$, and
	$\alpha = \max_{s_i, s_j} \frac{|s_i-s_j|}{2\sigma}$.
\end{restatable}

\parabf{Discussion}
Because the Gaussian mechanism does not satisfy pure DP,
we compare Bayes security with approximate DP.
For a function $f$ with sensitivity $\Delta f$,
and for $\varepsilon < 1$,
the following mechanism satisfies $(\varepsilon, \delta)$-DP:
$f(x) + \mathcal{N}(0, \frac{2\ln(\nicefrac{1.25}{\delta})(\Delta f)^2}{\varepsilon^2})$.

By applying \autoref{thm:beta-gaussians} we obtain
$\beta^* = 1-(\Phi(\alpha)-\Phi(-\alpha))$
with $\alpha = \nicefrac{\varepsilon}{2\sqrt{2\ln(\nicefrac{1.25}{\delta}}}$.
As desired, security does not depend on the sensitivity of the function.

We observe a similar behavior to what we observed for the Laplace mechanism (\autoref{fig:case-study} (c)).
Consider a dataset containing $N=1K$ records, for which an appropriate choice of
$\delta$ according to the literature is $\delta=\nicefrac{1}{N^2}$.
For a relatively secure setting ($\varepsilon=1$),
we have $\beta^*=0.925$.
As before, an interpretation of this value is that an optimal attacker will distinguish
the two most vulnerable secrets with probability $0.538$;
this is clearly non-negligible.
We note that only a stricter value such as
$\varepsilon = 0.1$ ensures a strong guarantee against the attack ($\beta^* = 0.992$).

Overall, Bayes security enabled us to interpret the privacy guarantees
of various mechanisms, by matching them back to the probability of success of an optimal attacker
under a specific threat model.

\section{Computational estimation of \minbeta}
\label{sec:fast-minimization}

Suppose that, differently from the cases we just analyzed (\autoref{sec:case-study}),
a simple closed-form expression of the mechanism does not exist:
how can we determine its Bayes security?
\autoref{thm:minbeta} shows that to quantify Bayes security, the minimizer
of multiplicative risk leakage $\beta$, we just need to estimate $\beta$
for all pairs of secrets; this requires $\mathcal{O}(n^2)$ measurements.
A measurement for a pair of
secrets is obtained by estimating the Bayes risk of
the mechanism for those two secrets;
we can do this analytically, if we have white-box knowledge
of the mechanism, or in a black-box manner\footnote{We remark that black-box estimation
of the Bayes risk (and, therefore, Bayes security) can be done consistently via distribution-free
techniques~\cite{cherubin2019fbleau}.}.
In either case, if the mechanism is complex enough (e.g., large
input or output space),
each measurement may need a non-negligible computational time,
from seconds to tens of minutes.

In this section, we investigate techniques for
improving the search time. Since the bottleneck is the
time it takes to measure $\beta(\priors, \channel)$ for one prior $\priors$,
we seek to reduce the number of such measurements.
We first assume white-box knowledge of the system
(subsections~\ref{sec:efficient-estimation}-\ref{sec:min-beta-approximation}),
and then study the black-box case (subsection~\ref{sec:minbeta-black-box}).

\parabf{Initial observations}
Denote by $\prior{ab}$ be the sparse prior vector
$(0, ..., 0, \nicefrac{1}{2}, 0, ..., 0, \nicefrac{1}{2}, 0, ..., 0)$
such that the two non-zero elements of value $\nicefrac{1}{2}$ are in
positions $a$ and $b$, and $a\neq b$.
Given a channel $\channel$, from the definition of $\beta$ we get that
\[
	\beta(\prior{ab}, \channel)
	= 2 -  \sum_o \max_{s\in\{a,b\}} \channel_{\secret, \object}
	~.
\]
The crucial observation (shown in the proof of \autoref{thm:beta-star-as-diameter}) is that
the above quantity is equal to the complement of the \emph{total variation}
distance $\TotalVar(\channel_a, \channel_b)$ between the rows $\channel_a$ and $\channel_b$ of the
channel.
The total variation distance of two discrete distribution is $\nicefrac{1}{2}$ of their
$\lnorm{1}$ distance (seen as vectors); hence:
\[
	\beta(\prior{ab}, \channel)
	~=~ 1 - \TotalVar(\channel_a, \channel_b)
	~=~ 1 - \frac{1}{2}\norm{\channel_a - \channel_b}_1
	~.
\]

Then, from \autoref{thm:minbeta}, we get that minimizing $\beta$ is
equivalent to finding the rows of the channel that are maximally distant
with respect to $\lnorm{1}$.
This is the well-known \emph{diameter problem} (for $\lnorm{1}$):
given the set of vectors $\channel_\secretspace$, find the two that are maximally
distant (i.e., find the diameter of the set).

\subsection{Computing \minbeta with domain knowledge} \label{sec:efficient-estimation}
In practical applications, domain knowledge may
enable \textit{a priori} identification of the two leakiest secrets.
For example, the smallest and largest webpages
users can visit in website fingerprinting (\autoref{sec:discussion});
and the smallest and largest exponents in timing side channels against exponentiation algorithms~\cite{cherubin2019fbleau}. %
There are also applications where all the secrets are equally vulnerable;
hence $\minbeta$ is obtained for any pair of distinct secrets.
For instance, when the mechanism operates in such a way
that all secrets enjoy the same protection (e.g.,
the Randomized Response mechanism, \autoref{sec:case-study}).

More generally, if one does not know the exact minimizing secrets,
but knows that they belong to a set $\secretspace' \subset \secretspace$,
then to determine $\minbeta$ it suffices measuring $\beta$ for all
$s_1, s_2 \in \secretspace'$.

\subsection{Computing \minbeta in linear time $\nsecrets$}
The geometric characterization given by \autoref{thm:beta-star-as-diameter}
implies that obtaining $\beta^*$ requires computing the diameter
of a set of $n = |\secretspace|$ vectors of dimension $m = |\objectspace|$.
The direct approach is to compute the distance between
every pair of vectors, i.e., perform $O(n^2 m)$ operations. This quadratic dependence on $n$ can be prohibitive
when the number of secrets grows.

We first show that, by using an isometric embedding of $\lnorm{1}^m$ into $\lnorm{\infty}^{2^m}$,
$\beta^*$ can be computed in $O(n 2^m)$ time.
Concretely, each $x \in \Reals^m$ is translated into a vector $\phi(x) \in \Reals^{2^m}$, which has
one component for every bitstring $b$ of length $m$,
such that $\phi(x)_b = \sum_{i=1}^{m} x_i (-1)^{b_i}$.
Note that the equivalence $\|\phi(x) - \phi(x')\|_\infty = \|x-x'\|_1$
holds for all $x,x'\in\Reals^m$.
The $\lnorm{\infty}$ diameter problem can be solved in linear time:
we only need to find the maximum and minimum value of each component.

This computation is linear in $|\secretspace|$ but exponential in $|\objectspace|$.
It outperforms the direct approach when the number of observations is small,
but the problem becomes harder as the number of observations grows. When $m = \Theta(n)$ there is no sub-quadratic algorithm for the $\lnorm{p}$-diameter problem for any $p\ge 0$~\cite{DBLP:journals/corr/DavidSL16} .
This suggests that there may not be any sub-quadratic time for computing $\beta^*$ either.

\subsection{An efficient approximation of \minbeta}
\label{sec:min-beta-approximation}

We present an estimation of $\minbeta$ that can be obtained in $O(nm)$ time.
One selects an arbitrary distribution $q\in\distset\objectspace$ and computes the
maximal distance $d$ between any channel row and $q$. The diameter of $\channel_\secretspace$
is at most $2d$, giving a lower bound on $\minbeta$.
Furthermore, if $q$ lies within the convex hull of $\channel_\secretspace$ (denoted by
$\ConvHull{\channel_\secretspace}$), then the diameter
is at least $d$, giving also an upper bound:

\begin{restatable}{proposition}{betastarapprox}
	\label{prop:beta-star-approx}
	Let $\channel$ be a channel,
	$q \in \distset{\objectspace}$, and
	$d = \max_{s\in \secretspace} \| \channel_s - q \|_1$.
	Then
	$1 - d \Wide\le \beta^*(\channel) \,.$
	Moreover, if $q \in \ConvHull{\channel_\secretspace}$ then
	$\beta^*(\channel) \Wide\le 1 - \nicefrac{d}{2}~.$
\end{restatable}

Good choices for $q$ are distributions that are likely to lie ``in-between''
the two maximally distant rows, for instance the \emph{centroid} of $\channel_\secretspace$
(mean of all rows). %

Several advanced approximation algorithms exist for the
$\lnorm{2}$ diameter problem \cite{DBLP:conf/cccg/ImanparastHM18}; these
could be employed using some embedding of $\lnorm{1}$ into $\lnorm{2}$. The
trivial embedding has distortion $\sqrt{m}$ (since $\|x\|_2 \le
\|x\|_1 \le \sqrt{m}\|x\|_2$), hence the approximation factor may be too
loose as $|\objectspace|$ grows. Low distortion embeddings of $\lnorm{1}$
into $\lnorm{2}$ exist \cite{DBLP:conf/stoc/AroraLN05}, but it is unclear
if they can be applied to the diameter problem.
In \autoref{appendix:beta-approximation}, we conduct an empirical
study of these approximations.

\subsection{Black-box estimation of $\beta^*$}
\label{sec:minbeta-black-box}

The previous sections assume full knowledge of
the channel $\channel$. In practice, this assumption may fail:
systems may be too complex
to analyze, or their behavior may be unknown.
In such cases,
we can estimate the Bayes risk, and therefore $\beta$, using black-box
estimation tools
(i.e., only observing the system's inputs and outputs),
such as
\texttt{F-BLEAU}~\cite{cherubin2019fbleau}.
As with the white-box case, we need to reduce the number
of priors $\pi$ for which we estimate $\beta(\prior, \channel)$.

\parait{Bounds} A first approach is to use the bounds given
by \autoref{prop:beta-star-approx}, which can be
computed in a black-box setting. One can interact
with the system to obtain observations for $q$. For instance observe:
the \textit{mean row}, by drawing observations from the channel with secrets that are chosen uniformly at random;
the \textit{any row of the channel}, by drawing observation for	a secret chosen arbitrarily;
or \textit{a row with arbitrary distribution}, e.g., by sampling $q$ uniformly at random from the set $\objectspace$.

\parait{Building upon $\bayesrisk$ black-box estimators~\cite{cherubin2019fbleau}}
If domain constraints do not enable identifying the pair of
leakiest secrets (\autoref{sec:efficient-estimation}),
we can try to reduce the search space.
For instance, we can
exploit the triangle inequality on the total variation distance
to discard some solutions before computing them.
E.g., given the Bayes security for
the priors $\prior{ac}$ and $\prior{bc}$:
\begin{align*}
\beta(\prior{ac}, \channel)+\beta(\prior{bc}, \channel)-1 &\leq
	\beta(\prior{ab}, \channel)\\
	&\leq 1 - |\beta(\prior{ac}, \channel)-\beta(\prior{bc, \channel})| \,.
\end{align*}
Thus, if $\beta(\prior{ab}, \channel)$ is larger than some
already-known $\beta(\prior{ij}, \channel)$ there is no need to compute it.
Conversely, if it is upper bounded by a small quantity,
we can compute it earlier aiming at discarding other combinations.

\section{Discussion and conclusions}
\label{sec:discussion}

This paper provides building blocks for studying complex
algorithms on the basis of Bayes security,
a metric that generalizes the cryptographic advantage.
Bayes security inherits benefits from both
average-case metrics, such as advantage and Bayes risk,
and worst-case
metrics, such as DP.
Similarly to the advantage, Bayes security is threat-specific:
it captures the risk for the users in a specified
threat model (e.g., what's the probability that a user's data record
is leaked).
Like DP, Bayes security is easily composable,
and it reflects the \textit{worst-case} for the two most vulnerable
secrets (e.g., data records).
Yet, Bayes security is a weaker worst-case notion than DP,
which may enable utility gains in high-security
regimes (\autoref{sec:case-study}).

\parabf{Applications}
The above characteristics make Bayes security
suitable for a broad range of security and privacy settings.
Below, we discuss some particularly fitting examples.

\parait{Website fingerprinting}
In website fingerprinting (WF), an adversary with access to
an encrypted network tunnel (e.g., VPN or Tor) aims to infer
the websites being visited by a user.
The \textit{success rate} (or \textit{accuracy}) of an attacker
has been used for years as a way of evaluating
an attack's goodness. However, this metric suffers from some
drawbacks~\cite{juarez2014critical,wang2020high}.
First, comparing success rate across studies is meaningless,
as the number of websites the user can visit strongly affects it:
the attack is very simple is the user is only allowed to visit
2 websites as opposed to 100.
Second, the prior probability of each website being visited
highly skews the success rate; if a website is easy to distinguish
from the others and it is very likely to be visited, then
the attacker's accuracy would be largely inflated.
The use of Mutual Information was suggested as an alternative metric~\cite{li2018measuring}.
However, Smith showed that this metric does not capture
the standard threat model used in WF, and it may
be misleading if we are ultimately interested in learning
about an attacker's success probability~\cite{smith2009foundations}.

$\beta$ was introduced for WF evaluation~\cite{cherubin2017website},
although without any theoretical justification.
In this work, we developed a theory for $\beta$,
and we showed that its minimizer,
the Bayes security metric, is particularly suited for WF:
i) it is prior independent; ii) it measures the risk
for the two leakiest secrets (i.e., the two websites that are
the easiest to tell apart);
iii) as shown in \autoref{img:ml-vs-beta-bound}, it is
captures particularly well the case of sparse prior
-- in WF, the prior over websites is highly sparse.
Overall, this suggests Bayes security is an appropriate choice
evaluating the user's risks against WF and, similarly,
the information leakage of WF defenses.
Future work may study if Bayes security implies bounds
w.r.t. other metrics of interest, such as True/False positives or
Precision and Recall (\autoref{sec:other-notions}).

\parait{PPML}
We suspect privacy preserving ML (PPML) algorithms can be easily
studied by using Bayes security.
Its strengths for this kind of analysis are:
i) it is easy to derive it analytically (e.g., as the total variation of
the posterior for the two leakiest secrets) (\autoref{sec:main-result});
ii) for a large secret space (e.g., data records in a dataset),
it characterizes the risk for the most vulnerable ones;
this, we argue, gives an easy interpretation of its guarantees;
iii) its prior independence helps studying mechanisms irrespective of the adversary's prior knowledge; and, once the attacker's prior
is known, it can be plugged in to better capture the risk
(\autoref{sec:other-notions});
iv) where an analytical study is not possible,
Bayes security can be easily estimated in a black-box manner (\autoref{sec:fast-minimization}).
Overall, we expect future work can provide Bayes security-style
guarantees for complex ML training pipelines.
For example, by exploiting our results on the Gaussian mechanism (\autoref{sec:case-study}),
it may be possible to study the security of DP-SGD
against common attacks such as membership inference~\cite{shokri2017membership},
attribute inference~\cite{fredrikson2015model},
and reconstruction~\cite{carlini2021extracting,balle2022reconstructing}.
This will enable bypassing bounds relating $\varepsilon$
and the advantage~\cite{yeom2018privacy,humphries2020differentially}, by computing the advantage (or Bayes security)
directly.
One immediate implication of \autoref{thm:minbeta} is
that evaluating membership inference attacks via the cryptographic
advantage (which, in this case, matches Bayes security),
gives guarantees for any prior probability that ``members''
may have.

\parait{Data release mechanisms}
Our analysis in \autoref{sec:case-study} suggests that, when
defending large datasets, Bayes security may help getting
better utility than DP in high-privacy regimes.

\parait{Fairness}
Bayes security captures the
risk for the most vulnerable pair of users (\autoref{thm:minbeta}).
We suspect this characteristic can be adapted
for evaluating privacy fairness (e.g., whether some population
subgroups enjoy better privacy than others).

\parabf{Further extensions}
In this paper, we discussed various extensions that
may further improve Bayes security's suitability to tackle
complex algorithms.
For example, proving a form of the \textit{miracle theorem}
(\autoref{sec:MultLeak}) would give analysts even further
flexibility when defining threat models for
real-world attacks.
Moreover, given the equivalence between Bayes security and
total variation (\autoref{thm:beta-star-as-diameter}),
it may be possible to exploit research on total variation
estimation to improve black-box leakage estimation
techniques.

In conclusion, Bayes security opens a new space in the security
metrics space, offering designers the opportunity to obtain
different trade-offs than previous metrics. As we showed in \autoref{sec:case-study},
these trade-offs enable the choice of security parameters that
provide strong protection and potentially with less
utility impact
under the threat model one chooses.

\section*{Acknowledgment}
The work of Catuscia Palamidessi has been funded by the European Research Council (ERC) grant Hypatia, grant agreement N. 835294.
We are grateful to Boris Köpf, Andrew Paverd, and Santiago Zanella-Beguelin for useful discussion.
We are especially thankful to Borja Balle and Lukas Wutschitz for proof-reading our manuscript, and for spotting the equivalence between Bayes security and a special case of approximate differential privacy.

\bibliography{biblio2-short}
\bibliographystyle{plain}

\appendix

\allowdisplaybreaks

\newenvironment{Reason}{\begin{tabbing}\hspace{2em}\= \hspace{1cm} \= \kill}
{\end{tabbing}\vspace{-1em}}
\newcommand\Step[2] {#1 \> $\begin{array}[t]{@{}llll}\displaystyle #2\end{array}$ \\}
\newcommand\StepR[3] {#1 \> $\begin{array}[t]{@{}llll}\displaystyle #3\end{array}$
\` {\RF \makebox[0pt][r]{\begin{tabular}[t]{r}``#2''\end{tabular}}} \\}
\newcommand\WideStepR[3] {#1 \>
$\begin{array}[t]{@{}ll}~\\\displaystyle #3\end{array}$ \`{\RF \makebox[0pt][r]{\begin{tabular}[t]{r}``#2''\end{tabular}}} \\}
\newcommand\Space {~ \\}
\newcommand\RF {\small}

\subsection{In general, $\beta$ is not minimized by the uniform prior}
\label{sec:bayes-inconsistency}

We start with the following lemma.

\begin{lemma}
	Suppose that $\beta(\priors, \channel) = 0$, for some system \system,
	and that \priors has $k$ non-zero components.
	Let $\priors' = (\nicefrac{1}{k}, ..., \nicefrac{1}{k}, 0, ..., 0)$, where the
	non-zero components are in correspondence of the non-zero components of $\priors$.
	Then we have $\beta(\priors', \channel) = 0$.
	\label{lemma:bayes-nonzero}
\end{lemma}

\begin{proof}
		Consider the k-dimensional simplex $Simp$ determined by the $k$ non-zero components of
		$\pi$.
		Since $\pi'$ has at most the same $k$ non-zero components, it is an element of $Simp$.
		Consider imaginary lines from $\pi'$ to each of the vertices of $Simp$.
		A vertex
		of $Simp$ is a vector of the form $(0,..., 0, 1, 0, ... , 0)$, i.e., one component is
		$1$ and all the others are $0$.
		Furthermore, the $1$ must be in correspondence of a non-zero component of $\pi$.
		These lines determine a partition of $Simp$ in convex subspaces, and $\pi$ must
		belong to one of them. Hence $\pi$ can be expressed as a convex combination of
		$\pi'$ and some vertices of $Simp$, say $\pi_1$, ..., $\pi_h$.
		Namely, $\pi = c \pi' + c_1 \pi_1 + ... + c_h \pi_h$ for suitable convex
		coefficients $c, c_1, ... , c_h$.
		Furthermore, since $\pi$ has $k$ non-zero components, it is an internal point of $Simp$,
		and therefore $c$ must be non-zero. Hence, we have:
		\begin{align*}
		0 &= \bayesrisk(\pi, \channel)\\
		  &= \bayesrisk(c \pi' + c_1 \pi_1 + ... + c_h p_h, \channel)\\
		  &\geq c \bayesrisk(\pi', \channel) + c_1\bayesrisk(\pi_1, \channel) + ... + c_h\bayesrisk(\pi_h, \channel)\\
		  &= c \bayesrisk(\pi', \channel) \,,
		\end{align*}
		where the third step comes from the concavity of $\bayesrisk$, and
		the last one is because $\bayesrisk(\pi_j, \channel) = 0, \quad \forall j$, since
		$\pi_j$ is a vertex.
		Therefore, since $c$ is not $0$, $\bayesrisk(\pi', \channel)$ must be $0$.
\end{proof}

We can now prove our result.

\begin{theorem}
	Let $\nsecrets = |\secretspace|$, and
	let $\unipriors$ denote the uniform prior on $\secretspace$.
	For any prior \priors with $k$ non-zero components, if
	$\beta(\priors, \channel) = 0$ then
	$$\beta(\unipriors, \channel) \leq \frac{1-\nicefrac{k}{\nsecrets}}{1-\nicefrac{1}{\nsecrets}} \,.$$
	Moreover, there exists a channel $\channel$ for which
	equality is reached.

\label{thm:bayes-inconsistency}
\end{theorem}
\begin{proof}
	Let $m = |\objectspace|$, and let $\secretspace'$
	be the set of the non-zero components of $\pi$.
	It is sufficient to note that:
	\begin{align*}
	&\sum_{o \in \objectspace} \channel_{s,o}\unipriors(s)\\
	&= \nicefrac{1}{\nsecrets}\sum_{o \in \objectspace} \max_{s \in S} \channel_{s,o}\\
	&\geq \nicefrac{1}{\nsecrets} \left(\channel_{s_1,o_1} + ... + \channel_{s_m,o_m}\right)
	\quad \text{where $s_i = \arg\max_{s \in S'} \channel_{s,o_i}$}\\
	&= \nicefrac{k}{\nsecrets} \,;
	\end{align*}
	the last equality is due to the fact that, by definition of $s_i$,
	$$\sum_o \max_{s\in S'} \channel_{so} =  (\channel_{s_1,o_1} + ... + \channel_{s_m o_m} ) \,.$$
	Therefore,  for $\priors'$ defined as in Lemma~\ref{lemma:bayes-nonzero},
	$\beta(\priors',C) = 0$ implies $\nicefrac{1}{k} (\channel_{s_1,o_1} + ... + \channel_{s_m o_m}) = 1$,
	from which we derive $(\channel_{s1 o1} + ... + \channel_{s_m o_m} ) = k$.
	This proves the first statement.

	The second claim of the theorem states the existence of a channel
	$\channel'$ for which equality is reached.
	We define $\channel'$ so that it coincides with \channel in the rows
	corresponding to the non-zero components of \priors.
	Define all the other rows identical to the previous ones (it does not matter which ones
	are chosen). Then:
	\begin{align*}
	\sum_o \max_{s \in \secretspace} \channel'_{s,o}
	\Wide=\sum_{o \in \objectspace} \max_{s \in \secretspace'} \channel'_{s,o}
	\Wide= \sum_{o \in \objectspace} \max_{s \in \secretspace'} \channel_{s,o}
	\Wide= k \,,
	\end{align*}
	therefore proving the second part of the theorem.
\end{proof}

\subsection{Proof of \autoref{thm:minbeta}}\label{sec:proofminbeta}
Let $\distsetcorner^{(k)}$, for $k=1, ..., n$,
be the set of priors with exactly $k$ non-zero components,
and such that the distribution on those components is uniform.
In the following we indicate with $\distsetcorner$ the
set $\distsetcorner = \distsetcorner^{(1)} \cup \distsetcorner^{(2)} \cup ... \cup \distsetcorner^{(n)}$.

We start by recalling the following definition from \cite{chatzikokolakis2008bayes}
(Definition 3.2, simplified).

\begin{definition}
	Let $S$ be a subset of a vector space, let $g: S \mapsto \mathcal{R}$,
	and let $S' \subseteq S$.
	We say that $g$ is \textit{convexly generated} by $S'$ if for all $v \in S$
	there exists $S'' \subseteq S'$ such that there exists a set of convex
	coefficients $\{c_u\}_{u \in S''}$ (i.e., satisfying
	$\sum_u c_u = 1$ and $c_u \geq 0$ $\forall u \in S''$)
	such that:
	\begin{align*}
	\text{1) }& v = \sum_{u \in S''} c_u u ~,
	&
	\text{2) }& g(v) = \sum_{u \in S''} c_u g(u) \,.
	\end{align*}
\end{definition}

The following results were also proven in the same reference (Proposition 3.9 in \cite{chatzikokolakis2008bayes}).

\begin{proposition}
	\label{prop:g-convexly-generated}
	$\errorguesspriors$ is convexly generated by $\distsetcorner$.
\end{proposition}

\begin{proposition}
	\label{prop:bayes-concave}
	$\bayesrisk(\pi,\channel)$ is concave on $\pi$.
\end{proposition}
The elements of $\distsetcorner$ are called \textit{corner points}
of $\errorguesspriors$, and the elements of each set
$\distsetcorner^{(k)}$ are the \textit{corner points of order $k$}.

We now prove that if a function is defined as the ratio
of a concave function and a convexly generated one, then
its minimum is attained on one of the
corner points of the function in the denominator.
This will be important to characterize the minimum of the
Bayes security metric, which is indeed defined as the ratio
of the Bayes risk  and the guessing error.

\begin{lemma}
	\label{lemma:min-of-ratio}
	Let $S$ be a subset of a vector space.
	Let $f: S \mapsto \mathcal{R}_{\geq 0}$ be a concave function,
	and let $g: S \mapsto \mathcal{R}_{\geq 0}$ be a function that
	is convexly generated by a finite $S' \subseteq S$,
	and which is positive in at least some of the elements of $S$.
	Then there exists $u \in S'$ such that
	$u = \argmin_{v: g(v)>0} \nicefrac{f(v)}{g(v)}$.
\end{lemma}

\begin{proof}
	Assume by contradiction that $\exists v \in S$ such that
	$g(v) > 0$ and $\forall u \in \distsetcorner$ with $g(u) > 0$
	\begin{equation}\label{eq:frac}
	\frac{f(v)}{g(v)} < \frac{f(u)}{g(u)} \,.
	\end{equation}

	Since $g$ is convexly generated by $S'$, and $g(v) > 0$,
	$\exists S'' \subseteq S'$ such that
	$g(v) = \sum_{u \in S''} c_u g(u)$, where $c_u$ are suitable
	convex coefficients, and $\forall u \in S''$
	$g(u) > 0$.
	Therefore:
	\[
	\begin{array}{rcll}
	\frac{f(v)}{g(v)} &=& \frac{f(v)}{\sum_{u \in S''} c_u g(u)}\\[2ex]
	&\geq &\frac{\sum_{u \in S''} c_u f(u)}{\sum_{u \in S''} c_u g(u)}\quad &\mbox{(by concavity of $f$)}\\[2ex]
	&> &\frac{\sum_{u \in S''} c_u g(u)\frac{f(v)}{g(v)}}{\sum_{u \in S''} c_u g(u)}  &\mbox{(by  \autoref{eq:frac})}\\[2ex]
	&= &\frac{f(v)}{g(v)}
	\end{array}
	\]
	which is impossible.
	Furthermore, $S'$ is finite, hence $\Big\{\frac{f(u)}{g(u)} \mid u \in S', g(v)>0\Big\}$ has a minimum.
\end{proof}

\begin{corollary}\label{cor:corner}
	The minimum of
	$\{\beta(\priors, \channel) | \priors \in \distset{\secretspace}, \errorguesspriors(\priors) > 0\}$
	exists, and it is in one of the corner points of $\errorguesspriors$.
\end{corollary}

\begin{proof}
	The statement follows from the definition  $\beta(\priors, \channel) = \frac{\bayesrisk(\priors, \channel)}{\errorguesspriors(\priors)}$,
	and from \autoref{prop:g-convexly-generated}, \autoref{prop:bayes-concave},
	and
	\autoref{lemma:min-of-ratio}.
\end{proof}

Furthermore, note that because $\errorguesspriors$ in its
corner points of order 1 takes value $0$, the corner point of
$\errorguesspriors$ on which $\beta$ is minimized
must have order $k \geq 2$.

We now can prove \autoref{thm:minbeta}.
It remains to show that the corner points
on which $\beta$ is minimized have order $k=2$.

\ThmMinBeta*
\begin{proof}
	We show this result by induction over $n$, the cardinality of $\secretspace$,
	where we assume $n \geq 2$.\\

	\textit{Base case ($n=2$).} Since $\forall \priors \in \distsetcorner^{(1)}$
	the guessing error is $\errorguesspriors(\priors ) = 0$,
	the minimizer has to have order $2$.\\

	\textit{Inductive step.} Assuming we proved the result for $n$,
	we  prove it for $n+1$.
	By \autoref{cor:corner}, it is sufficient to show that:
	\begin{equation}\label{eq:reduction}
	\forall \priors \in \distsetcorner^{(n+1)} \; \exists \priors' \in \distsetcorner^{(n)}
	\;
	\;
	 \frac{\bayesrisk(\priors, \channel)}{\errorguesspriors(\priors)} \geq
	\frac{\bayesrisk(\priors', \channel)}{\errorguesspriors(\priors')} \,.
	\end{equation}

	Consider the $(n+1)\times m$ channel matrix $\channel$.
	For each row $i$, we define $p_i$ as the sum of elements
	of the row which are the maximum in their column. (Ties are broken arbitrarily.)
	I.e.,
	$$p_i \; \stackrel{\rm def}{=}\; \sum_{o} \channel_{i,o}\, I(\channel_{i,o} = \max_{s} \channel_{s,o}) \,.$$
	where $I(S)$ is the indicator function, i.e., the function that gives $1$ if the statement $S$ is true, and $0$ otherwise.

	Similarly, we define $q_i$ as the sum of elements
	which are the second maximum in the columns that have maximum in column $i$.
	More precisely,
	let ${\rm smax}(A)$ be the function returning the second maximum
	in a set $A$; for instance, if $a_1 \geq a_2 \geq a_3 \geq \ldots$,
	then ${\rm smax}(\{a_i\}) = a_2$.
	Again, ties are broken arbitrarily.
	Then:  $$q_i \;\stackrel{\rm def}{=} \; \sum_{o} \channel_{j,o}\, I(\channel_{i,o} = \max_{s} \channel_{s,o}  \; \mbox{and} \; \channel_{j,o} = {\rm smax}_{s} \channel_{s,o} ) \,.$$
	Note that the elements that compose $q_i $ are in rows different from $i$ and possibly different from each other.

	Without loss of generality, assume that we have:
	\begin{equation}
	\label{eq:main-proof-assumption}
	p_{n+1}-q_{n+1} \; = \; \min_i (p_i - q_i )\,.
	\end{equation}
	We further denote by $r_o$, for $o=1, \ldots, k$,
	the elements of the $(n+1)$-th row that are not the components of $p_{n+1}$, namely
	 \[
	\{  r_o \,|\, o = 1,\ldots , k\} \; \stackrel{\rm def}{=} \; \{ \channel_{n+1,o} \,|\, \channel_{n+1,o}  \neq  \max_{s} \channel_{s,o}\}
	\]
	The following observation is immediate:
	\begin{fact}
	\label{fact:main-proof-ineq}
	For all $i \in \{1, ..., n+1\}$, we have $q_i \geq r_i$.
	\end{fact}

	We can now prove \autoref{eq:reduction}. We will prove it for
	$\priors' = (\nicefrac{1}{n}, ..., \nicefrac{1}{n}, 0)$, while
	$\priors = (\nicefrac{1}{(n+1)}, ..., \nicefrac{1}{(n+1)})$ necessarily.

	Observe that:
	\begin{align*}
	&\bayesrisk(\priors, \channel) = 1 - \frac{1}{n+1}\sum_{i=1}^{n+1} p_i =
		\frac{n+1-\sum_{i=1}^{n+1} p_i}{n+1}\\
	&\errorguesspriors(\priors) = \frac{n}{n+1}\\
	&\bayesrisk(\priors', \channel) = 1 - \frac{1}{n}\sum_{i=1}^{n} p_i - q_{n+1} =
		\frac{n - \sum_{i=1}^{n} p_i - q_{n+1}}{n}\\
	&\errorguesspriors(\priors') = \frac{n-1}{n} \,.
	\end{align*}

	Therefore, to prove \autoref{eq:reduction} we need to demonstrate that:
	$$(n-1)\left(n+1 - \sum_{i=1}^n p_i - p_{n+1}\right) \geq
		n\left(n-\sum_{i=1}^n p_i - q_{n+1}\right) \,.$$

	By simplifying and rearranging:
	$$\sum_{i=1}^n \, p_i - n\, p_{n+1} + n\, q_{n+1} + p_{n+1} \geq 1 \,.$$

	\noindent By the assumption in \autoref{eq:main-proof-assumption}, we have:
	\begin{Reason}
	\Step{}{
		\sum_{i=1}^n p_i \;-\; &n\, p_{n+1} + n\,q_{n+1} + p_{n+1}\\
	}
	\Step{$\geq$}{
		\sum_{i=1}^n p_i - \sum_{i=1}^n p_i + \sum_{i=1}^n q_i + p_{n+1}\\
	}
	\StepR{$=$}{\autoref{fact:main-proof-ineq}}{
		\sum_{i=1}^n q_i + p_{n+1}\\
	}
	\StepR{$\geq$}{$\channel$ is stochastic}{
		\sum_{i=1}^n r_i + p_{n+1} \Wide= 1 ~.
	}
	\end{Reason}
\end{proof}

\subsection{Proofs of \autoref{sec:fast-minimization}}
\label{appendix:fast-minimization-proofs}

\ThmBetaStarAsDiameter*
\begin{proof}
	Denote by $\prior{ab}$ the prior assigning probability $1/2$ to $a,bin\secretspace, a\neq b$.
	We show that
	\begin{equation}\label{eq6442}
		\beta(\prior{ab}, \channel)
		= 2 - \sum_o \max_{s\in{a,b}} \channel_{\secret, \object}
		= 1 - \frac{1}{2}\|\channel_a - \channel_b\|_1
		~,
	\end{equation}

	Denote $\channel_{\uparrow,o} = \max_{s\in{a,b}} \channel_{\secret, \object}$
	and $\channel_{\downarrow,o} = \min_{s\in{a,b}} \channel_{\secret, \object}$.
	The fact that
		$\beta(\prior{ab}, \channel)
		= 2 - \sum_o \channel_{\uparrow,o}$ comes directly from the definition
		of $\beta$.

	Since $\channel_a$ and $\channel_b$ are probability distributions, it holds that
	\[
		\sum_o (\channel_{\uparrow,o} + \channel_{\downarrow,o}) = 2
	\]
	Hence, we have that
	\[
		\| \channel_a - \channel_b\|_1
		= \sum_o (\channel_{\uparrow,o} - \channel_{\downarrow,o})
		= 2 \sum_o \channel_{\uparrow,o}  - 2
	\]
	and \eqref{eq6442} follows directly.
	We conclude by \autoref{thm:minbeta}.
\end{proof}

We now show that taking convex combinations of vectors cannot increase the diameter of a set,
which will be useful for both \autoref{prop:beta-star-approx} and \autoref{thm:cascade-comp}.
Denote by $\Diam{S}$, $\ConvHull{S}$ the diameter and the convex hull of
$S$ respectively.

\begin{lemma}\label{lem:diameter-convex-hull}
	For any $S \subseteq \Reals^n$, it holds that
	\[
		\Diam{\ConvHull{S}} \Wide= \Diam{S},
	\]
	where distances are measured wrt any norm $\|\cdot\|$.
\end{lemma}
\begin{proof}
	Let $d = \Diam{S}$. Since $S \subseteq \ConvHull{S}$ we clearly have
	$d \le \Diam{\ConvHull{S}}$,
	the non-trivial part is to show that
	$d \ge \Diam{\ConvHull{S}}$.

	We first show that
	\begin{equation}\label{eq4352}
	\forall a\in S, b \in \ConvHull{S} : \|a - b\| \le d
	~.
	\end{equation}
	Let $a \in S,b\in \ConvHull{S}$ and denote by $B_d[a]$ the closed
	ball of radius $d$ centered at $a$.
	The diameter of $S$ is $d$, hence
	\begin{align*}
		B_d[a] &\supseteq S ~,
		&\text{and since balls are convex} \\
		B_d[a] = \ConvHull{B_d[a]} &\supseteq \ConvHull{S}
		~,
	\end{align*}
	which implies $\|a - b\| \le d$.

	Finally we show that
	\[
		\forall b,b'\in\ConvHull{S} : \|b - b'\| \le d
		~.
	\]
	Let $b,b'\in\ConvHull{S}$, from \eqref{eq4352} we know
	that $B_d[b] \supseteq S$, and since balls are convex we have that
	$B_d[b] \supseteq \ConvHull{S}$, which implies $\|b-b'\| \le d$.
\end{proof}

\betastarapprox*
\begin{proof}
	Let $s,s'\in\secretspace$, from the triangle inequality we have that
	\[
		\|\channel_s-\channel_s'\|_1 \Wide\le \|\channel_s-q\|_1 + \|\channel_s'-q\|_1
		~,
	\]
	hence $\Diam{\channel_\secretspace} \le 2 d$;
	\autoref{thm:beta-star-as-diameter} implies the lower bound.

	Moreover, assume that $q\in\ConvHull{\channel_\secretspace}$.
	Since $\channel_s \in \ConvHull{\channel_\secretspace}$ for all
	$s\in\secretspace$,
	it holds that
	\[
		\Diam{\ConvHull{\channel_\secretspace}}
		\Wide\ge
		\max_{s\in\secretspace}\|\channel_s - q\|_1
		\Wide= d
		~.
	\]
	From \autoref{lem:diameter-convex-hull} we get that
	\[
		\Diam{\channel_\secretspace}
		\Wide=\Diam{\ConvHull{\channel_\secretspace}}
		\Wide\ge d
		~,
	\]
	which gives us an upper bound
	from \autoref{thm:beta-star-as-diameter}.
\end{proof}

\subsection{Proofs of \autoref{sec:composing}}
\label{appendix:composing-proofs}

\ThmParallelComp*
\begin{proof}
	Recall that $\pi_{ab} \in \distset\secretspace$ denotes the
	prior that assigns probability $1/2$ to both $a,b\in\secretspace, a\neq b$.
	We will use the fact that for such priors, $\beta(\pi_{ab},\channel)$ can be written as
	\begin{equation}\label{eq0934}
		\beta(\pi_{ab},\channel)
		\Wide=
		\sum_{o \in \objectspace} \min_{s \in \{a,b\}} \channel_{s,o}
		~.
	\end{equation}
	This comes from the definition of $\beta$ and the fact that the rows $\channel_a$ and $\channel_b$
	are probability distributions, hence
	\[
		(\sum_{o \in \objectspace} \min_{s \in \{a,b\}} \channel_{s,o}) +
		(\sum_{o \in \objectspace} \max_{s \in \{a,b\}} \channel_{s,o}) =
		\sum_{o \in \objectspace} \sum_{s \in \{a,b\}} \channel_{s,o} =
		2
		~.
	\]

	We also use the basic fact that
	for non-negative $\{q_i,r_i\}_i$:
	\begin{equation}
		\label{eq4093}
		(\min_i q_i) \cdot (\min_i r_i)
		\Wide\le
		\min_i (q_i \cdot r_i)
	\end{equation}

	The proof proceeds as follows:
	\begin{Reason}
	\Step{}{
		\minbeta(\ChA) \minbeta(\ChB)
	}
	\StepR{$=$}{\autoref{thm:minbeta}}{
		\big( \min_{a,b \in \secretspace} \beta(\pi_{ab}, \ChA) \big) \big( \min_{a,b \in \secretspace} \beta(\pi_{ab}, \ChB) \big)
	}
	\StepR{$\le$}{\eqref{eq4093}}{
		\min_{a,b \in \secretspace} \big(\beta(\pi_{ab}, \ChA)\cdot \beta(\pi_{ab}, \ChB) \big)
	}
	\StepR{$=$}{\eqref{eq0934}}{
		\min_{a,b \in \secretspace}
		\big( \sum_{o_1 \in \objectspace^1}  \min_{s\in\{a,b\}} \ChA_{s,o_1} \big)
		\big( \sum_{o_2 \in \objectspace^2}  \min_{s\in\{a,b\}} \ChB_{s,o_2} \big)
	}
	\WideStepR{$=$}{Rearranging sums, distributively}{
		\min_{a,b \in \secretspace}
		\sum_{o_1 \in \objectspace^1}\sum_{o_2 \in \objectspace^2}
		\big(\min_{s\in\{a,b\}} \ChA_{s,o_1} \big)
		\big(\min_{s\in\{a,b\}} \ChB_{s,o_2})\big)
	}
	\StepR{$\le$}{\eqref{eq4093}}{
		\min_{a,b \in \secretspace}
		\sum_{o_1 \in \objectspace^1}\sum_{o_2 \in \objectspace^2}
		\min_{s\in\{a,b\}} \ChA_{s,o_1}  \ChB_{s,o_2}
	}
	\StepR{$=$}{Def. of $\ChA || \ChB$}{
		\min_{a,b \in \secretspace} \sum_{o\in\objectspace^1\times\objectspace^2}
		\min_{s\in\{a,b\}} (\ChA||\ChB)_{s,o}
	}
	\StepR{$=$}{\eqref{eq0934}}{
		\min_{a,b \in \secretspace} \beta (\pi_{ab}, \ChA || \ChB)
	}
	\Step{$=$}{
		\minbeta(\ChA || \ChB)
	}
	\end{Reason}

	This bound is tight. E.g., $\minbeta(\channel || \channel) = \minbeta(\channel)\cdot\minbeta(\channel)$ for
	$$\channel = \begin{bmatrix}
		0.4 & 0.6\\
		0 & 1\\
	\end{bmatrix} \,.$$
\end{proof}

\ThmCascadeComp*
\begin{proof}
	The
	$\minbeta(\channel^1\channel^2) \geq \minbeta(\channel^1)$
	part is easy, and comes from the fact that
	\[
		\bayesrisk(\pi, \channel^1\channel^2) \Wide\geq \bayesrisk(\pi, \channel^1)
	\]
	for all priors $\pi$, hence also for the one achieving $\minbeta(\channel^1\channel^2)$.

	The more interesting part is to show that
	$\minbeta(\channel^1\channel^2) \geq \minbeta(\channel^2)$.
	The key observation is that the rows of $\channel^1\channel^2$ are convex combinations
	of those of $\channel^2$:
	\[
		(\channel^1\channel^2)_{s_1} = \sum_{s_2 \in \secretspace^2} \channel^1_{s_1, s_2} \channel^2_{s_2}
		\qquad \text{for all } s_1 \in \secretspace^1
	\]
	Denote by $\channel_\secretspace =  \{ C_{s}\ |\ s \in \secretspace \}$
	the set of $\channel$'s rows, we have:
	\begin{align*}
		(\channel^1\channel^2)_{\secretspace^1}
		&\Wide\subseteq \ConvHull{\channel^2_{\secretspace^2}}
			~,
			&\text{hence} \\
		\Diam{(\channel^1\channel^2)_{\secretspace^1}}
		&\Wide\le \Diam{\ConvHull{\channel^2_{\secretspace^2}}}
		~.
	\end{align*}
	Finally, from \autoref{lem:diameter-convex-hull}
	we get that
	\[
		\Diam{(\channel^1\channel^2)_{\secretspace^1}}
		\Wide\le \Diam{\ConvHull{\channel^2_{\secretspace^2}}}
		\Wide= \Diam{\channel^2_{\secretspace^2}}
		~.
	\]
	\autoref{thm:beta-star-as-diameter} concludes,
	since
	$
		\minbeta(\channel) = 1- \frac{1}{2} \Diam{\channel_\secretspace}
	$.
\end{proof}

\subsection{Proofs of \autoref{sec:other-notions}}
\label{appendix:proofs-other-notions}

\ThmBoundOnBayes*
\begin{proof}
	From \autoref{thm:minbeta} we know that for every \channel there exists a $\priors^*$  with a support containing only two secrets, and uniformly distributed on them, such that
	$\minbeta(\priors^*) = \min_\priors \beta(\priors,\channel)$.
	Given \channel, let us assume, without loss of generality,  that the two secrets are $s_1$ and $s_2$, and that for each $o$ in the first $k$ columns we have $\channel_{s_1,o}\geq \channel_{s_2,o}$, and in the last $m-k$ columns we have $\channel_{s_1,o}< \channel_{s_2,o}$.
	Then, if we define
	\begin{equation}
	\label{eqn:abdef}
	a \stackrel{\rm def}{=} \sum_{j=1}^{k}  \channel_{s_1,o_j} \quad \mbox{and} \quad b \stackrel{\rm def}{=}  \sum_{j=k+1}^{m}  \channel_{s_2,o_j},
	\end{equation}
	we have $1-b \leq a $ and $1-a < b$. From the constraints \eqref{eqn:LDP}, we also know that $a\leq \exp(\varepsilon)(1-b)$ and $b \leq \exp(\varepsilon) (1-a)$.
	Hence there exists $x,y$ with $1\leq x \leq \exp(\varepsilon) $ and $1< y \leq \exp(\varepsilon) $  such that $a = x(1-b)$ and $b = y(1-a)$.
	The figure below illustrates the situation  in the first two rows of the matrix:
	\begin{figure}[h]
		\begin{center}
			\begin{tabular}{|c|c|}
				\hline
				$\quad a \; = \;  x\,(1-b) \quad$ & $\quad 1- a\quad$ \\
				\hline
				$\quad 1-b \quad$  & $\quad  b \; = \;  y\,(1- a) \quad$ \\
				\hline
			\end{tabular}
		\end{center}
	\end{figure}

	From $a = x(1-b)$ and $b = y(1-a)$ we derive
	\begin{equation}\label{eqn:aandb}
	a = \frac{xy-x}{xy -1} \quad \mbox{and} \quad b = \frac{xy-y}{xy -1} ,
	\end{equation}
	from which we can compute the Bayes risk of $\channel$ in $\priors^*$, as a function of $x$ and $y$:
	\begin{eqnarray}
	f(x,y) &\stackrel{\rm def}{=} &\bayesrisk(\priors^*,\channel) \\
	&= &1 - \pi^*_{s_1} a - \pi^*_{s_2}b \\
	&= &1 - \frac{1}{2}(a+b) \\
	&= &\frac{\frac{1}{2}(x+y)-1}{xy-1}.\label{eqn:br}
	\end{eqnarray}
	In order to find the minimum of $f(x,y)$, we compute its partial derivatives:
	\begin{align*}
		\frac{\partial f}{\partial x}(x,y) &=
			\frac{-\frac{1}{2}y^2 + y -\frac{1}{2}}{(xy-1)^2} ~,
		&
		\frac{\partial f}{\partial y}(x,y)
		&= \frac{-\frac{1}{2}x^2 + x -\frac{1}{2}}{(xy-1)^2} ~.
	\end{align*}
	Since $1\leq x \leq \exp(\varepsilon) $ and $1< y \leq \exp(\varepsilon) $, we can easily see that  both partial derivatives are negative, hence the minimum value of $\bayesrisk(\priors^*, \channel)$ is obtained for the highest possible values of $x$ and $y$, namely $x=y=\exp(\epsilon)$.
	Therefore, using ~\eqref{eqn:br}:
	\begin{equation}\label{eqn:brr}
	\bayesrisk(\priors^*, \channel)  \;  \geq \; \frac{\frac{1}{2}(\exp(\epsilon)+\exp(\epsilon))-1}{\exp(\epsilon)\exp(\epsilon)-1}  \; = \; \frac{1}{\exp(\epsilon)-1}.
	\end{equation}\\
	Finally, since $G(\priors^*) = \nicefrac{1}{2}$, we conclude:
	$$
	\begin{array}{rclr}
	\beta(\priors,\channel) &= &\displaystyle\frac{\bayesrisk(\priors,\channel)}{G(\priors)} \\
	&\geq &\displaystyle\frac{\bayesrisk(\priors^*,\channel)}{G(\priors^*)} &\;\;\;\;\mbox{(by \autoref{thm:minbeta})}\\
	&\geq &\displaystyle\frac{2}{1+\exp(\varepsilon)}&\;\;\;\;\mbox{(by \eqref{eqn:brr})}.
	\end{array}
	$$

	\noindent
	\ref{LDP2}
	Consider the values $a,b$   defined   in \eqref{eqn:abdef}. Any $n\times m$ matrix $\minchannel$ for which these values are
	$$a \; = \;  b \; = \; \frac{\exp(\epsilon)}{\exp(\epsilon) +1 }$$
	achieves the above lower bound  for $\beta$:
	$$\minbeta({\minchannel}) \; = \;  \beta(\pi^*,\minchannel) \; = \;  \frac{2}{1+\exp(\varepsilon)}$$
	Hence the bound is a minimum.
	\autoref{fig:min channels} shows two examples of such matrices.
\end{proof}

\subsection{Generalization using gain/loss functions}
\label{appendix:gain-loss}

We describe here the generalizations of $\mulleakage$ and $\beta$,
parameterized by a \emph{gain function} $g$ (for vulnerability) or a
\emph{loss function $\ell$} (for risk),
discussed in \autoref{sec:other-notions}.

Let $\guessspace$ be the set of \emph{guesses} the adversary
can make \emph{about} the secret; a natural choice is $\guessspace=\secretspace$, but
other choices model a variety of adversaries, (e.g., guessing
a part or property of the secret, or making an approximate guess).
A \emph{gain} function $g(w,s)$ models the adversary's gain when guessing $w\in\guessspace$ and
the actual secret is $s\in\secretspace$. Prior and posterior $g$-vulnerability \cite{alvim2012measuring}
are the \emph{expected gain of an optimal guess}:
\begin{align*}
	V_g(\pi) &= \textstyle\max_w \sum_s \pi_s g(w,s) ~,
	&
	V_g(\pi, \channel) &= \textstyle\sum_o p(o) V_g(\delta^o) ~,
\end{align*}
where $\delta^o$ is the posterior
distribution on $\secretspace$ produced by the observation $o$.
Then $g$-leakage expresses how much vulnerability increases due to the channel:
$\mulleakage_g(\pi,\channel)=V_g(\pi, \channel) / V_g(\pi)$.

Similarly, we use a \emph{loss} function $\ell(w,s)$, modelling
how the adversary's loss in guessing $w$ when the secret is $s$.
Prior and posterior $\ell$-risk are the expected
loss of an optimal guess:
\begin{align*}
	R_\ell(\pi) &= \textstyle\min_w \sum_s \pi_s \ell(w,s) ~,
	&
	R_\ell(\pi, \channel) &= \textstyle\sum_o p(o) R_\ell(\delta^o) ~.
\end{align*}
Then $\beta_\ell$ can be defined by comparing prior and posterior
risk: $\beta_\ell(\pi,\channel) =R_\ell(\pi, \channel) / R_\ell(\pi)$.

Clearly, $\mulleakage = \mulleakage_g$ for the
\emph{identity} gain function, given by
$g(w,s) = 1$ iff $w=s$ and $0$ otherwise. Similarly, $\beta = \beta_\ell$
for the 0-1 loss function, $\ell(w,s) =0$ iff $w=s$ and $0$ otherwise.

\subsection{Proofs of \autoref{sec:case-study}}

We first prove results on the Bayes security of Gaussian mechanisms, which serves
as a template for the security derivation for the Laplace distribution.

\ThmGaussian*

Before determining the Bayes security of this mechanism, we prove a simple lemma:

\begin{lemma}[Total variation of two Gaussians]
	\label{thm:tv-gaussians}
	$$\TotalVar(\mathcal{N}(\mu_p, \sigma^2), \mathcal{N}(\mu_q, \sigma^2)) =
	\Phi(\alpha) - \Phi(-\alpha)$$
	where $\Phi$ is the CDF of $\mathcal{N}(0, 1)$
	and $\alpha = \frac{|\mu_p-\mu_q|}{2\sigma}$.
\end{lemma}

\begin{proof}
	This proof follows closely the proof of Theorem 3.1 in \cite{li19}.
	Let us first calculate the total variation between the following two Gaussians:
	$P \sim \mathcal{N}(\mu_p, 1)$ and $Q \sim \mathcal{N}(\mu_q, 1)$.
	Without any loss of generality, let
	$\mu_p \leq \mu_q$;
	hence, $P(x) \geq Q(x)$ for
	$x \leq \frac{\mu_p+\mu_q}{2}$.

	\begin{align*}
		\TotalVar(P, Q)
		&=\int_{-\infty}^{\frac{\mu_p+\mu_q}{2}}
		P(x)-Q(x) dx\\
		&= \int_{-\frac{\mu_p-\mu_q}{2}}^{\frac{\mu_p-\mu_q}{2}} P(x)-Q(x) dx\\
		&=P_{A\sim\mathcal{N}(0, 1)}\left(A \in \left[-\frac{\mu_p-\mu_q}{2}, \frac{\mu_p-\mu_q}{2}\right]\right)\\
		&=\Phi\left(\frac{\mu_p-\mu_q}{2}\right) -
		\Phi\left(-\frac{\mu_p-\mu_q}{2}\right)
	\end{align*}

	Now, observe that the total variation is scale-independent.
	Hence we have: $\TotalVar(\mathcal{N}(\mu_p, \sigma^2), \mathcal{N}(\mu_q, \sigma^2)) = \TotalVar(\mathcal{N}(\nicefrac{\mu_p}{\sigma}, 1), \mathcal{N}(\nicefrac{\mu_q}{\sigma}, 1))$.
	Applying the result obtained above to the scaled means concludes the proof.
\end{proof}

\ThmLaplace*

\begin{proof}[Sketch proof]
	We first compute the total variation distance between
	$\Lambda(\mu_p, 1)$ and $\Lambda(\mu_q, 1)$.
	By applying similar arguments to the one used for \autoref{thm:tv-gaussians},
	we obtain:
	$$\TotalVar(\Lambda(\mu_p, 1), \Lambda(\mu_q, 1)) = F(\alpha) - F(-\alpha) \,,$$
	where $F$ is the CDF of $\mathcal{L}(0,1)$, and
	$\alpha = \frac{|\mu_p-\mu_q|}{2}$.

	Via explicit computation of $F$:
	$$\TotalVar(\Lambda(\mu_p, 1), \Lambda(\mu_q, 1)) = 1-\exp(-\alpha) \,.$$

	We apply the above result to compute the total variation between
	$\Lambda(\nicefrac{\mu_p}{\lambda}, 1)$ and $\Lambda(\nicefrac{\mu_q}{\lambda}, 1)$
	(which is equal to the total variation between
	$\Lambda(\mu_p, \lambda)$ and $\Lambda(\mu_q, \lambda)$),
	and use the equivalence between total variation and Bayes security
	metric to conclude the proof.
\end{proof}

\subsection{Empirical evaluation of Randomized Response}
\label{appendix:rr}

\parabf{Dataset}
The US 1990 Census dataset (\census) comprises 2,458,285 records with a number of attributes
for each record. As Murakami and Kawamoto~\cite{Murakami019},
we reduced the attributes to those we judged potentially sensitive: age (8 values), income (5 values),
marital status (5 values), and sex (2 values).
Overall, the number of values that can be taken by the
vectors describing an individual is $8 \times 5 \times 5 \times 2 = 400$.
This is the size of both the secret and output spaces for RR.

\parabf{Methodology}
First, we use RR to obfuscate the \census dataset to
guarantee different levels of $\varepsilon$-LDP,
and measure the resulting utility by computing
the total variation distance between the empirical estimation $\hat{p}$
and the true distribution $p$.

Second, we compute the Bayes security for RR,
both analytically using \autoref{eq:minbetaRR} for different number of secrets $\nsecrets = |\secretspace|$,
and empirically using \texttt{fbleau}~\cite{cherubin2019fbleau}.
Because the Bayes security between any pair
of secrets obfuscated by RR is identical, one computation
for one arbitrary pair suffices to obtain \minbeta.

\begin{figure}
	\centering
		\includegraphics[width=0.45\columnwidth]{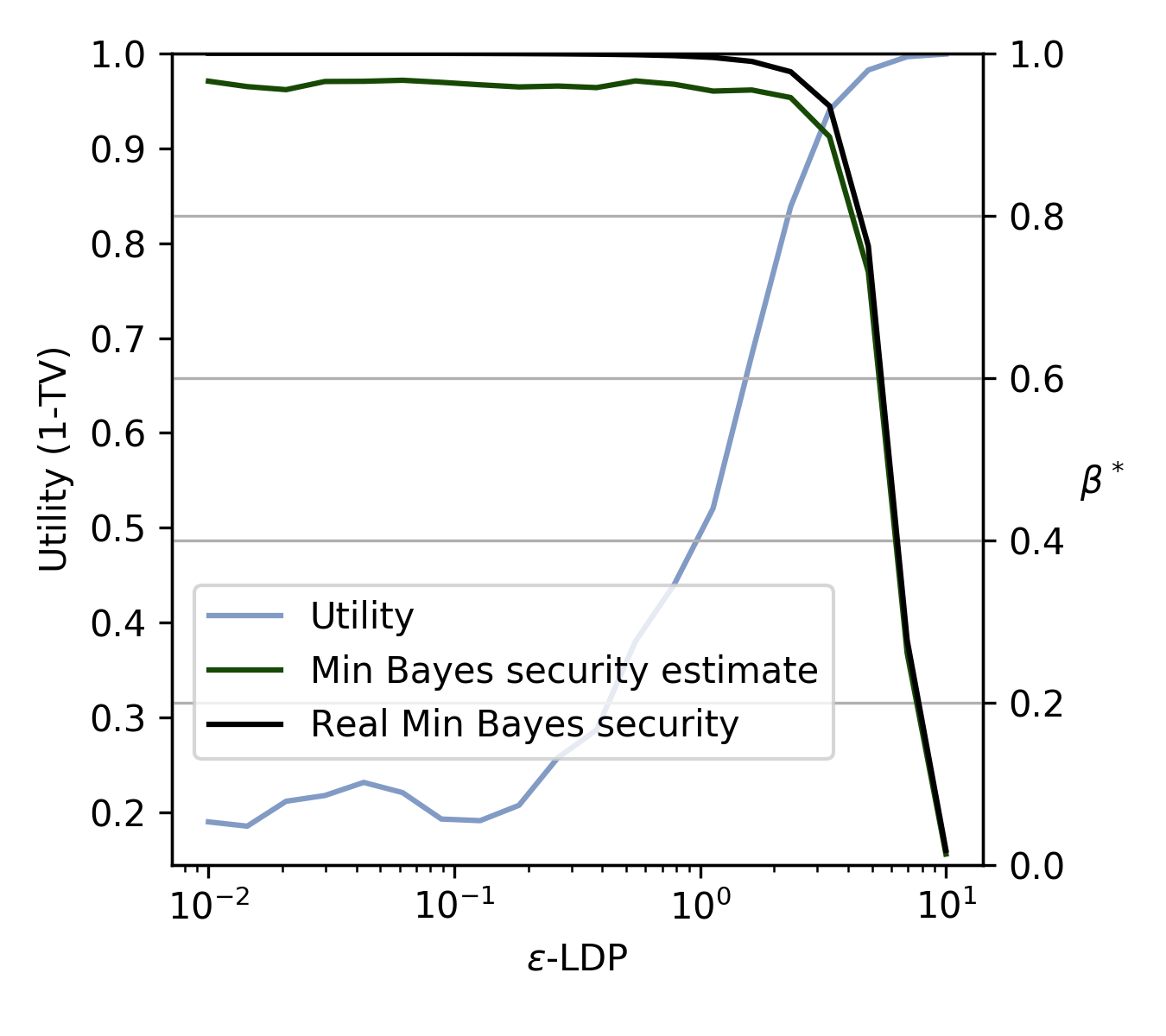}
	\caption{Randomized Response obfuscation mechanism on the US \census dataset.
	Utility (blue line) vs security measured in Min Bayes security (green line) and DP (x-axis).
	}
	\label{fig:rr}
\end{figure}

\parabf{Results} We show in \autoref{fig:rr}
the security of the empirical estimation $\hat{p}$
(Bayes security, empirical and estimation, in green, and DP in the x-axis),
and the utility after applying RR to obtain $\varepsilon$-LDP for
$\varepsilon \in [0.01, 10]$. We observe that
utility is low for values $\varepsilon < 2$.
Concretely, utility reaches 95\% for $\varepsilon = 3.3$.
While this is weak protection in terms of differential privacy,
we obtain $\minbeta=0.96$ \texttt{fbleau} estimate
($\minbeta=0.99999$ from \autoref{eq:minbetaRR}),
which means that the adversary's probability
of success is small. Even for $\varepsilon = 4.8$, which
yields utility of $99\%$, \minbeta is above $0.9$
($\minbeta=0.99995$ analytic value).

\subsection{Bayes security approximation via \autoref{prop:beta-star-approx}}
\label{appendix:beta-approximation}

\begin{figure}[ht!]
	\centering
	\includegraphics[width=0.75\linewidth]{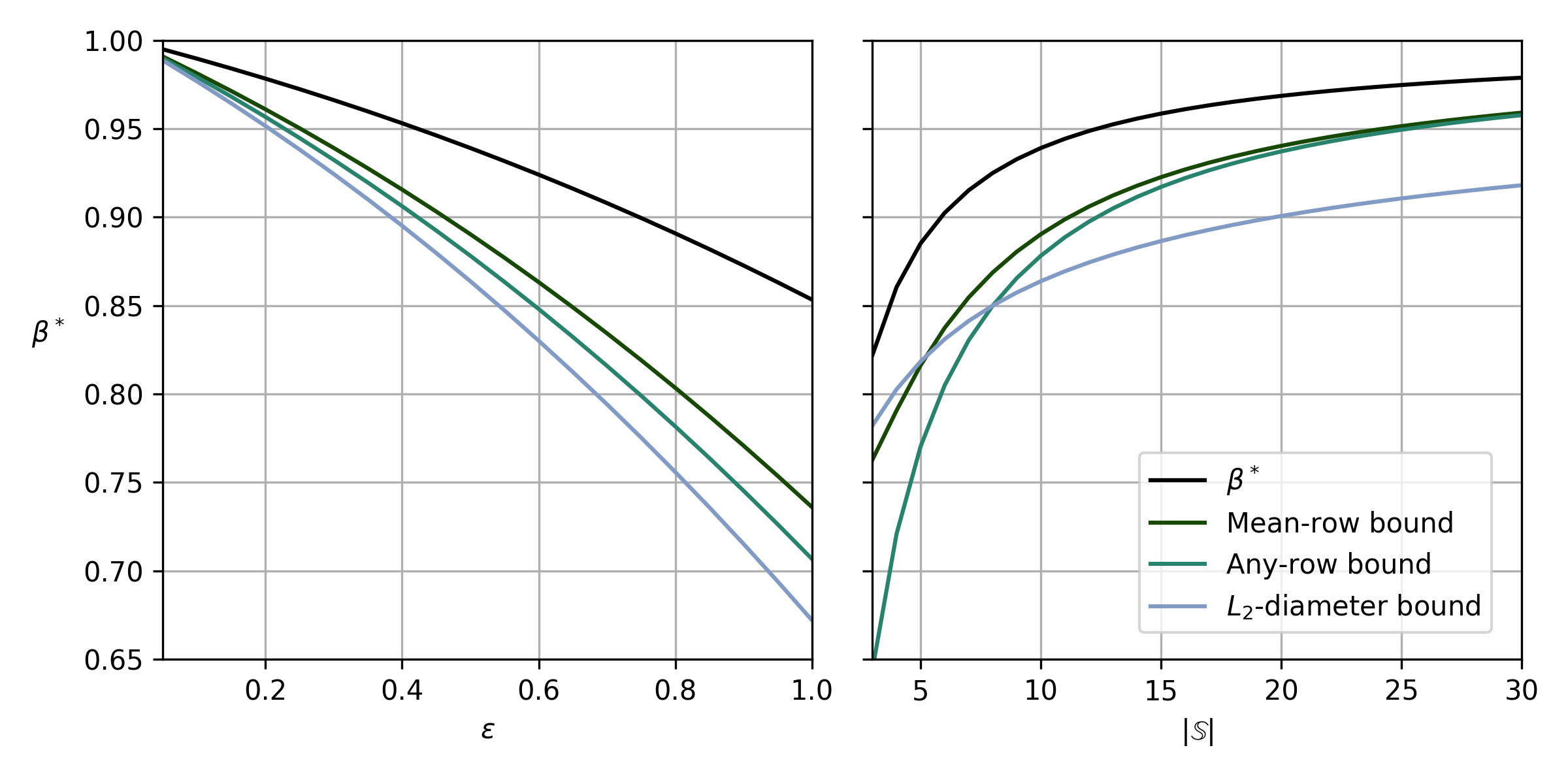}
	\caption{$\minbeta$ and bounds for Rand. Response. Left: $|\secretspace|=|\objectspace|=10$ while
		varying $\varepsilon$. Right: $\varepsilon = 0.5$ while varying $|\secretspace|=|\objectspace|$.}
	\label{fig:rr-bounds}
\end{figure}

In \autoref{fig:rr-bounds}, we show $\minbeta$ and various lower
bounds for the Randomized Response mechanism (RR) (\autoref{sec:case-study}). Two of the bounds are obtained via
\autoref{prop:beta-star-approx}, by setting $q$ to be the mean row
(the uniform distribution for RR)
and any row of $\channel$ (all rows are equal for RR). The third bound is obtained
via the $\lnorm{2}$-diameter by using the trivial embedding.
The mean-row bound is the most accurate
in this case.

\end{document}